\documentclass[11pt,onecolumn]{article}  
\usepackage{times}
\usepackage{amsmath}
\usepackage{amssymb}
\usepackage{xspace}
\usepackage{theorem}
\usepackage{graphicx}
\usepackage{ifpdf}
\usepackage{url,hyperref}
\usepackage{latexsym}
\usepackage{euscript}
\usepackage{xspace}
\usepackage{color}
\usepackage{makeidx}
\usepackage{picins,wrapfig}

\long\def\remove#1{}

\newtheorem{theorem}{Theorem}[section] % section
\newtheorem{lemma}[theorem]{Lemma}

\newtheorem{obs}[theorem]{Observation}
\newtheorem{corollary}[theorem]{Corollary}
\newenvironment{proof}{{\em Proof:}}{\hfill{\hfill\rule{2mm}{2mm}}}

\def\marrow{{\marginpar[\hfill$\longrightarrow$]{$\longleftarrow$}}}

\newcommand{\erin}[1] {{\sc Erin says: }{\marrow\sf [#1]}}

\newcommand {\mm}[1] {\ifmmode{#1}\else{\mbox{\(#1\)}}\fi}
\newcommand{\denselist}{\itemsep 0pt\parsep=1pt\partopsep 0pt}

\newcommand{\area}	{\rm {\mm Area}}

\newcommand{\reals}	{{\rm I\!\hspace{-0.025em} R}}
\newcommand{\morph}	{homotopy\xspace}

\newcommand{\energy}	{cost\xspace}
\newcommand{\similarity}	{similarity\xspace}

\newcommand{\simC}	{\sigma}

\newcommand{\breakpt}	{anchor point\xspace}
\newcommand{\wn}	{\mm {\rm wn}}
\newcommand{\totalW}	{total winding number\xspace}
\newcommand{\totalwn}	{\mm {\rm Tw}}

\newcommand{\acurve}	{\gamma}
\newcommand{\acell}		{R}

\newcommand{\arr}			{\mm {\rm Arr}}

\newcommand{\uc}	{\mathcal{U}}
\newcommand{\lift}[1]	{{\tilde{{#1}}}}
\newcommand{\concatenate}		{\circ}
\newcommand{\sphere}	{\mathbf{S}} %{\mathbb{S}}

\newcommand{\sphereminus}[1]   {\sphere_{#1}}

\newcommand{\zerobf}		{\mathbf{0}}

\newcommand{\B}			{\mathrm {\mathbf{B}}}
\newcommand{\Z}			{\mathrm {\mathbb{Z}}}
\newcommand{\z}			{\mathbf{z}}
\newcommand{\w}			{\mathbf{w}}

\newcommand{\bpt}		{\mathbf{b}}
\newcommand{\mysquare}		{\Box} %{\mathbf{D}}

\newcommand{\amorph}		{H}
\newcommand{\Ipt}		{\mathbf{x}}
\newcommand{\optmorph}		{H^*}

\newcommand{\p}			{\mathbf{p}}
\newcommand{\q}			{\mathbf{q}}
\newcommand{\pin}		{\p}
\newcommand{\qin}		{\q}

\newcommand{\myproofbegin}	{{\noindent {\em Proof:~}}}
\newcommand{\myproofend}   {\hfill{\hfill\rule{2mm}{2mm}}}
\newcommand{\sensepreserving}	{sense-preserving\xspace}
\newcommand{\allpositive}	{consistent\xspace}
\newcommand{\extendQ}	{\hat{Q}}
\newcommand{\extIpt}		{\hat{\Ipt}} %{\Ipt} %{\hat{\Ipt}}
\newcommand{\representative}	{representative\xspace}
\newcommand{\addwn}		{addW} %{\mathbf{W}}
\newcommand{\myparagraph}[1]	{\vspace*{0.1in}{\noindent {\bf {#1}~}}}
\newcommand{\anothermorph}	{H}
 
\newcommand{\Wmin}		{wn-min\xspace}
\newcommand{\Wmax}		{wn-max\xspace}
\newcommand{\valid}		{valid\xspace}

\newcommand{\Frechet}   {Fr\'{e}chet\xspace}
\newcommand{\etal}      {et al.\@\xspace}

\newcommand{\trimesh}		{K}
\newcommand{\tricomplexity}		{N}

\begin{document}

\title{Measuring Similarity Between Curves on 2-Manifolds via Homotopy Area}
\author{Erin Wolf Chambers \thanks{Dept. of Math and Computer Science, Saint Louis University, Saint Louis, MO. } \and 
%Research partially supported by NSF grant CCF 1054779.} \and
Yusu Wang \thanks{Dept. of Computer Science and Engineering, The Ohio State University, Columbus, OH 43210.}}
%\\
%Dept. of Computer Science and Engineering \\
%The Ohio State University, Columbus, OH }
\date{}
\maketitle

\begin{abstract}
Measuring the similarity of curves is a fundamental problem arising in many application fields.  There has been considerable interest in several such measures, both in Euclidean space and in more general setting such as curves on Riemannian surfaces or curves in the plane minus a set of obstacles.  
%However, so far, only limited results are known, especially for the case when the underlying domain is a general surface. 
However, so far, efficiently computable similarity measures for curves on general surfaces remain elusive. 
This paper aims at developing a natural curve similarity measure that can be easily extended and computed for curves on general orientable $2$-manifolds. Specifically, we measure similarity between homotopic curves based on how hard it is to deform one curve into the other one continuously, and define this ``hardness'' as the minimum possible surface area swept by a homotopy between the curves. We consider cases where curves are embedded in the plane or on a triangulated orientable surface with genus $g$, and we present efficient algorithms (which are either quadratic or near linear time, depending on the setting) for both cases. 
%The results are also extended to comparing simple cycles (simple closed curves) in these settings, although the algorithms become less efficient  by a factor of $n$ or $n^2$, depending on the setting.
%Let $n$ be the total number of vertices from input curves, and $I$ the number of intersections between input curves. The running time of our algorithm depends near-quadratically on $I$, and is near linear in $n$ if $I$ is a constant. 

\end{abstract}

\section{Introduction}
Measuring curve similarity is a fundamental problem arising in many application fields, including graphics, computer vision, and geographic information systems. Traditionally, much research has been done on comparing curves embedded in the Euclidean space. However, in many cases it is natural to study curves embedded in a more general space, such as a terrain or  a surface. 

In this paper, we study the problem of measuring curve similarity on surfaces. Specifically, given two simple homotopic curves embedded on an orientable $2$-manifold (including the plane), we measure their similarity by the minimum total area swept when deforming one curve to the other (the ``area" of the homotopy between them), and present efficient algorithms to compute this new measure. 

\myparagraph{Related work.}
From the perspective of computational geometry, the most widely studied similarity measures for curves is the \Frechet{} distance. Intuitively, imagine that a man and his dog are walking along two paths with a leash between them. The \Frechet{} distance between these two paths is the minimum leash length necessary for them to move from one end of the paths to the other end without back-tracking. Since the \Frechet{} distance takes the ``flow" of the curves into account,  in many settings it is a better similarity measure for curves than alternatives
such as the Hausdorff distance \cite{ag-dgsmi-00,akw-cdmpc-04}. 

Given two polygonal curves $P$ and $Q$ with $n$ total edges in $\reals^d$, the \Frechet{} distance can be computed in $O(n^2\log n)$ time \cite{ag-cfdbt-95}. An $\Omega(n\log n)$ lower bound for the decision  problem in the algebraic computation tree model is known \cite{bbkrw-hdwd-07}, and Alt has conjectured that the decision problem is 3SUM-Hard~\cite{alt2009}.  Recently, Buchin \etal{}~\cite{bbmm-fswd-12} show that there is a real algebraic decision tree to solve the \Frechet{} problem with sub-quadratic depth, suggesting that perhaps this is not the case. They also give an improved algorithm which runs in $O(n^2 \sqrt{\log n} (\log \log n)^2)$ time.
Very recently, Agarwal \etal{} present a novel approach to compute the discrete version of the \Frechet{} distance between two polygonal curves in sub-quadratic time \cite{AAKS13}. This is the first algorithm for any variant of the \Frechet{} distance to have a sub-quadratic running time for general curves. No previous algorithm, \emph{exact or approximate}, with running time $o(n ^2)$ is known for general curves, although sub-quadratic approximation algorithms for special families of curves are known \cite{akw-cdmpc-04,ahkww-fdcr-06,dreimelharpeled2012}. 

While the \Frechet{} distance is a natural curve similarity measure, it is sensitive to outliers. Variants of it, such as the summed-\Frechet{} distance, and the partial \Frechet{} similarity, have been proposed \cite{BBW09,Buchin07,EFV07}, usually at the cost of further increasing the time complexity. 

The problem of extending and computing the \Frechet{} distance to more general metric space has also received much attention. Geodesic distance between points is usually considered when the underlying domain is not $\reals^d$. For example, Maheshwari and Yi \cite{MY05} computed the geodesic \Frechet{} distance between two polygonal paths on a convex polytope in roughly $O(n^3K^4\log (Kn))$ time, where $n$ and $K$ are the complexity of the input paths and of the convex polytope, respectively.  Raichel and Har-Peled consider approximating the weak \Frechet{} distance between simplicial complexes in $\mathbb{R}^d$~\cite{harpeledraichel2011}.
Geodesic \Frechet{} distance between polygonal curves in the plane within a simple polygon has also been studied~\cite{b-ombp-02,CW08,ehgmm-nsmpa-02}. 

%So far, only limited results are known for extending \Frechet{} similarity measures to surfaces (for example, the \Frechet{} distance was only extended to convex polytopes~\cite{MY05} previously). 
%Computing (variants of) the \Frechet{} distance also seems to typically induce high computational complexity for more general domains. 

Rather than comparing distance between only two curves, Buchin et.~al.~\cite{bbklsww-mt-10} propose the concept of a median in a group of curves (or trajectories, in their setting).  They give two algorithms to compute such a median.   The first is based simply on the concept of remaining in the middle of the set of curves; this algorithm, while fast and simple, has a drawback in that the representative curve might not capture relevant features shared by a majority of the input curves.  Their second algorithm addresses this issue by instead isolating a subset of  relevant curves which share the same homotopy type with respect to obstacles that are placed in empty regions of the blame; it then computes a medial curve from this relevant subset.

One issue with generalizing \Frechet{} distance directly to surfaces is that the underlying topology is not taken into account; for example, in geodesic \Frechet{} distance, while the length of the leash varies continuously, the actual leash itself does not.  
As a result, several measures of similarity have been proposed which take the underlying topology into account. 
Chambers \etal{} \cite{CVE08} proposed the so-called homotopic \Frechet{} distance and gave a polynomial (although not efficient) algorithm for when the curves reside in a planar domain with a set of polygonal obstacles. The extra requirement for this homotopic \Frechet{} distance is that the leash itself and not just its length has to vary in a continuous manner, essentially restricting the homotopy class which the leash is in.  A stronger variant called isotopic \Frechet{} distance has also been proposed and investigated, although no algorithms at all are known to even approximate this distance~\cite{Chambers_isotopicfrechet}.

Orthogonal to homotopic \Frechet{} distance is the concept of the \emph{height} of a homotopy;  instead of minimizing the maximum leash length, this measure views the homotopy as tracing a way for the first curve to deform to the second curve, where the goal is to minimize the longest intermediate curve length.  Introduced independently in two very different contexts~\cite{Brightwell2009, chambersletscher2009}, it is not even known if the problem is in NP. 

Recent work on approximating the homotopy height and the homotopic \Frechet{} distance has yielded efficient $O(\log n)$ approximation algorithms for both of these problems~\cite{nomagicleash2012}. However, exact algorithms on surfaces for either problem are still unknown.

\myparagraph{New work. }
In this paper, we develop a natural similarity measure for curves on general surfaces that can be computed both quickly and exactly. Intuitively, we measure distances between homotopic curves based on how hard it is to deform one curve into the other one, and define this ``hardness'' as the minimal total surface area swept by a homotopy between them, which we call the optimal homotopy area. 
Our similarity measure is natural, and robust against noise (as the area in a sense captures average, instead of maximum, deviation from one curve to the other). To the best of our knowledge, this is the first similarity measure for curves on general surfaces with efficient polynomial-time algorithms to compute it exactly. 

It is worth noting that this definition in a way combines homotopic \Frechet{} distance with homotopy height; those measures compute the ``width" and ``height" of the homotopy, while our measure calculates the total area.  It is thus interesting that while no exact algorithms are known for either of those measures on surfaces, we are able to provide a polynomial running time for computing the area of a homotopy. 

We consider both cases where curves are embedded in the plane, or on a closed, triangulated orientable surface with genus $g$. For the former case, our algorithm runs in $O(n \log n + I^2 \log I)$ time, where $n$ is the total complexity of input curves and $I$ is the number of intersections between them. 
%The running time is \emph{near linear} in $n$ if there are a constant number of intersection points. 
On a surface, if the input  is a triangulation of complexity $\tricomplexity$, then our algorithm runs in time $O(I^2\log I + ng\log n + \tricomplexity)$. 
While our similarity measure is more expensive to compute for the case of  curves in the plane than the \Frechet{} distance when $I = \omega(n)$, one  major advantage is that this measure can be computed on general orientable surfaces efficiently. In fact, the ideas and algorithms behind the planar case form the foundation for the handling of the case on general surfaces.  
%Finally, these results can also be extended to comparing two simple cycles (closed curves), although the time complexity then adds an extra linear factor for the planar case, and quadratic factor for the surface case. 

The main ideas behind our approach are developed by examining some properties of one natural class of homotopies, including a relation with the  winding number of a closed curve. Specifically, the use of the winding number enables us to compute the optimal homotopy area  efficiently in the plane, where the homotopy is restricted to be piecewise differential and regular. This forms the basis of our  dynamic programming framework to compute similarity between curves in the plane. We also show how to build efficient data structures to keep the total cost of the dynamic program low.

For the case where the underlying surface is a topological sphere, we extend the winding number in a natural way and show how to adapt our planar algorithm without additional blow-up in the time complexity. 
For the case when the  surface has non-zero genus, we must extend our algorithm to run efficiently in the universal cover (which is homeomorphic to the plane) by using only a small portion of it. 

We remark that the idea of measuring deformation areas has been used before in practice \cite{Cro92,MS92}. For example, similarity between two convex polygons can be measured by their symmetric difference~\cite{afrw-mcsrs-98,Vel01}; we note that this is not equivalent to homotopy area, although it may be the same value in some situations.  In another paper, the area sandwiched between an \emph{$x$-monotone curve} and another curve is used to measure their similarity~\cite{BC06}. However, computing the ``area'' between general curves has not been investigated prior to this work. 

\section{Definitions and Background}

\myparagraph{Paths and cycle.}
We will assume that we are working on an orientable 2-manifold $M$ (which could be the plane). A \emph{curve} (or a \emph{path}) on a surface $M$ is a map $P: [0,1] \rightarrow M$; a \emph{cycle} (or a \emph{loop}) is a continuous map $\gamma: S^1 \rightarrow M$ where $S^1$ is the unit circle. A curve $P$ or a cycle $\gamma$ is \emph{simple} if $P(t_1) \neq P(t_2)$ (resp. $\gamma(t_1) \neq \gamma(t_2)$) for any $t_1 \neq t_2$.  

\myparagraph{Homotopy}
A \emph{homotopy} between two paths $P$ and $Q$ (with the same endpoints) is a continuous map $H: [0,1] \times [0,1] \rightarrow M$ where $H(0, \cdot) = P$, $H(1, \cdot) = Q$, $H(\cdot, 0) = P(0) = Q(0)$ and $H(\cdot, 1) = P(1) = Q(1)$.  
%Similarly, a \emph{free homotopy} between two cycles $\gamma$ and $\delta$ is a continuous map $H: [0,1] \times S^1 \rightarrow M$ such that $H(0, \cdot) = \gamma$ and $g(1, \cdot) = \delta$.
A homotopy describes a continuous deformation between the two paths or curves: for any value $t \in [0,1]$, we let $\amorph_t = H(t, \cdot)$ be the intermediate curve at time $t$, where $\amorph_0 = P$ and $\amorph_1 = Q$.

We define the \emph{area} of a homotopy $H$ to be the total area covered by the image of the homotopy on the surface, where an area that is covered multiple times will be counted with multiplicity.   More precisely, given a homotopy $H$ whose image is piecewise differentiable, \[ \text{Area}(H) = \int_{s \in [0,1]} \int_{t \in [0,1]} \left\lvert \frac{dH}{ds} \times \frac{dH}{dt} \right\rvert ds dt \].  The \emph{minimum homotopy area} between $P$ and $Q$ is the infimum of the areas of all homotopies between $P$ and $Q$, denoted by $\simC(P,Q)$.  
%, since there are cases where the function will not be integrable, depending on the input curves, the underlying space, and the specific homotopy between them.  
If such an infimum does exist and can be achieved by a homotopy, we call that homotopy an \emph{optimal homotopy}. 
%we call the homotopy which produces the minimum area the \emph{optimal homotopy}. 
%If one exists, a homotopy that produces the minimum area is an \emph{optimal homotopy}.

We note that it is not immediately clear that this value exists, depending on the curves and underlying homotopy. 
Minimum area homotopies were considered by Douglas~\cite{douglas} and Rado~\cite{rado} in the context of Plateau's problem; they noted that not only is the integral improper in general, but the infimum itself may not be continuous.   The eventual proof that these exist in $\mathbb{R}^n$ relies on a definition using Dirichlet integrals which ensure (almost) conformal parameterizations of the homotopy.  See the book by Lawson~\cite{lawson1980lectures} for an  overview of this result as well as several extensions to minimal area submanifolds in more general settings.  

\parpic[r]{\includegraphics[width=3.5cm]{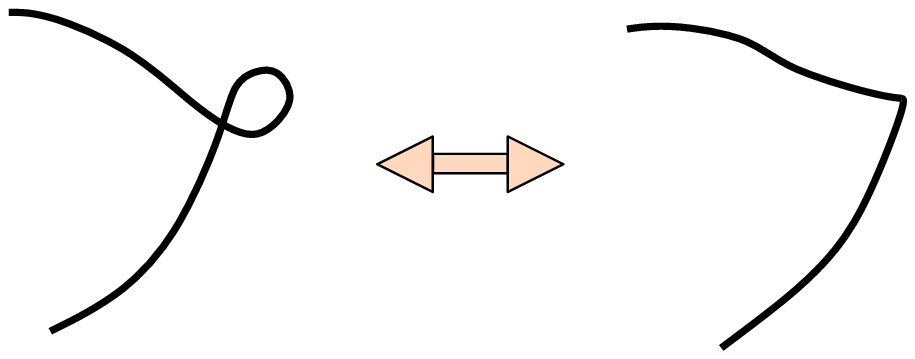}}
However, beyond a proof of existence, we are interested in computing such homotopies, or at least measuring their actual area, in much simpler settings such as $\mathbb{R}^2$ or a  surface.  
To this end, we restrict the input curves to be \emph{simple} curves which consist of a finite number of piecewise analytic components. We also need $H$ to be continuous and piecewise differentiable, so that the integral can be defined. Finally, we will also require that at any time $t$, the intermediate curve $H_t$ is \emph{regular} (see \cite{Whi37} for smooth curves and \cite{MY91} for piecewise-linear curves). 
Intuitively, this means that the deformation is ``kink''-free \cite{MY91}, and cannot create or destroy a local loop as shown in the right figure (the singular point in the right curve is a kink).   
Note that this is required for the minimum homotopy to even exist; again we refer the reader to the book by Lawson~\cite{lawson1980lectures} for details.

\myparagraph{Decomposing arrangements.}
Consider two simple piecewise analytic curves $P$ and $Q$ with the same endpoints.  Their concatenation forms a (not necessarily simple) closed curve denoted by $C = P \concatenate \text{rev}(Q)$, where rev$(Q)$ is the reversal of $Q$.  
Let $\arr(C)$ denote the arrangement formed by $C$, where vertices in $\arr(C)$ are the intersection points between $P$ and $Q$. 
%(Note that any vertices of the curves of $P$ and $Q$ are ignored here.)
An edge / arc in $\arr(C)$ is a subcurve of either $P$ or $Q$. 
%For sake of exposition, the illustrations in this paper usually draw $Q$ as a horizontal segment (see the right figure). 

\parpic[r]{\includegraphics[width=3.5cm]{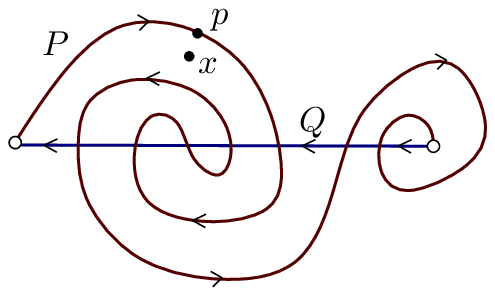}}
We give $C$ (and thus $P$ and $Q$) an arbitrary orientation. Hence we can talk about the \emph{sidedness} with respect to $C$ at a point $p \in P$.  Specifically, a point $x \in \reals^2$ is to the \emph{right} of $C$ at $p$ if it is a counter-clockwise turn from the orientation of the vector $p x$ the orientation of (tangent of) $C$ at $p$ (see the right figure for an example).  
Given two oriented curves $\gamma_1$ and $\gamma_2$, an intersection point $p$ of them is \emph{positive} if it is a counter-clockwise turn from the orientation of $\gamma_1$ to that of $\gamma_2$ at $p$.  
%(Note that this also agrees with our crossing definition in the context of winding numbers.)
For a curve $\gamma$ and a point $x \in \gamma$, the \emph{index of x} is the parameter of $x$ under the arc-length parameterization of $\gamma$. We sometimes use $x$ to represent its index along $\gamma$ when its meaning is clear from the context. 
Given two points $x,y \in \gamma$, we will use $\gamma[x,y]$ to denote the unique sub-curve of $\gamma$ between points $x$ and $y$. 

We say that a \morph{} $\amorph$ from $P$ to $Q$ is \emph{right \sensepreserving{}} if for any $t, s \in [0,1]$, we have that either $\amorph_{t+dt}(s) = \amorph_t(s)$ or $\amorph_{t+dt}(s)$ is to the right of the oriented curve $\amorph_t$ at $\amorph_t(s)$. If it is the former case, then we say that $p = \amorph_t(s)$ is a \emph{fixed point} at time $t$. 
Similarly, we say that $\amorph$ is \emph{left \sensepreserving{}} if for any $t,s \in [0,1]$, $\amorph_t(s)$ is either a fixed point or deforms to the \emph{left} of the curve $\amorph_t$. 
Our homotopy $\amorph$ is \emph{\sensepreserving{}} if it is either right or left \sensepreserving{}. 
The \sensepreserving{} property means that we can continuously deform the curve $P$ always in the same direction, without causing local folds in the regions swept. Intuitively, any optimal \morph{} should have this property to some extent, which we will make more precise and prove later.  

%%%%%%%%%%%%%%%%%%%%%%%%%%%%
\section{Structure of Optimal Homotopies}
\label{sec:optimal}

Given two simple curves $P$ and $Q$ (with the same end points) embedded on an orientable $2$-manifold $M$, let $X = \{ \Ipt_1, \ldots, \Ipt_I \}$ denote the set of $I$ intersection points between them, sorted by their order along $P$. 
Given a homotopy $H$ from $P$ to $Q$, a point $p \in M$ is called an \emph{\breakpt{} with respect to $\amorph$} if it remains on $H(t, \cdot) = \amorph_t$ at all times $t \in [0,1]$. Of course not all intersection points are anchor points.   However, if $p$ is an \breakpt{}, then it is necessarily an intersection point between $P$ and $Q$, as $p \in \amorph_0 = P$ and $p \in \amorph_1 = Q$. We exclude the beginning and ending end points of $P$ and $Q$ from the list of \breakpt{}s, as they remain fixed for all homotopies.
In what follows, we show that any optimal \morph{} can be decomposed by \breakpt{}s such that each of the resulting smaller homotopies has a simple structure. 

Specifically, consider an arbitrary optimal \morph{} $\optmorph$. Let $\B = \{ \bpt_1, \ldots, \bpt_k \}$ be the set of \breakpt{}s with respect to $\optmorph$, the minimum area homotopy. We order the $\bpt_i$'s by their indices along $P$. It turns out that the order of their indices along $Q$ is the same, and the proof of this simple observation is in Appendix \ref{appendix:breakptorder}. 
\begin{obs}
The order of $\bpt_i$'s along $P$ and along $Q$ are the same. 
\label{obs:breakptorder}
\end{obs}

This observation implies that we can decompose $\optmorph$ into a list of sub-homotopies, where $\optmorph_i$ %: [\pin_i, \pin_{i+1}] \times [0,1] \rightarrow M$ 
morphs $P[\bpt_i, \bpt_{i+1}]$ to $Q[\bpt_i, \bpt_{i+1}]$. Obviously, each $\optmorph_i$ is necessarily optimal, and it induces no \breakpt{}s. The following result states that an optimal \morph{} without \breakpt{}s has a simple structure, which is \sensepreserving{}. Intuitively, if any point changes its deformation direction at any moment, the deformation will sweep across some area redundantly and thus cannot be optimal. The detailed proof is in Appendix \ref{appendix:sensepreserving}. 

\begin{lemma}
If an optimal \morph{} $\anothermorph$ from $P$ to $Q$ has no \breakpt{}s, then it is \sensepreserving{}. 
\label{lem:sensepreserving}
\end{lemma}

%%%%%%%%%%%%%%%%%%%%%%%%%%%%
\section{Minimum Area Homotopies In The Plane}
\label{sec:plane}

In this section, we consider the case where the input consists of two simple polygonal curves in the plane. We develop an algorithm to compute the \similarity{} between $P$ and $Q$ in $O(I^2 \log I + n \log n)$ time, where $n$ is the total complexity of input curves and $I$ is the number of intersections.  Note that $I = \Theta(n^2)$ in the worst case, although of course it may be much smaller in some cases. Although efficient algorithms for comparing curves in the plane exist (such as  the \Frechet{} distance), our planar algorithm will be the fundamental component for comparing curves on general surfaces in the next section. 
It turns out that our approach can easily be extended to measure similarity between simple cycles in the plane; see Appendix \ref{appendix:sec:planarcycles}. 

%%%%%%%%%%%%%%%%%%%%%%%%%
\subsection{Relations to Winding Numbers}
\label{sec:curvesinplane}

We are given two simple (open) curves in the plane which share common endpoints. 
Previously, we have shown that if an optimal \morph{} does not induce \breakpt{}s, then it is \sensepreserving{}. The implication of this result is manifested by using the \emph{winding number}, defined for a loop $\gamma$ in the plane at a base point $x$. 

Intuitively, imagine starting from a point $y$ on $\gamma$, and connecting $x$ and $y$ by a string. The winding number at $x$ w.r.t. $\gamma$, denoted by $\wn(x; \gamma)$, is an integer measuring how many times this string winds, in a clockwise manner, around $x$ as $y$ traverses $\gamma$. 
More formally, consider an infinite ray $f$ based at $x$ which is generic (so it has a finite set of transversal intersections / crossings with $\gamma$).  
%An intersection point $\gamma(t)$ between a curve $r$ and $\gamma$ is a \emph{crossing} if for all sufficiently small $\epsilon$, the points $\gamma(t-\epsilon)$ and $\gamma(t+\epsilon)$ lie on opposite sides of $r$.  
Consider a crossing $\gamma(t)$ between the ray $f$ and $\gamma$. 
This crossing is \emph{positive} if the triangle  $x$, $\gamma(t)$, and $\gamma(t+\epsilon)$ is oriented counterclockwise, 
and is \emph{negative} if oriented clockwise. 
Then $\wn(x; \gamma)$ is the number of positive crossings minus the number of negative crossings with respect to any generic ray from $x$. 

%\parpic[r]{\includegraphics[width=2.7cm]{./figs/wnexample1} \hspace*{0.05in} \includegraphics[width=2.5cm]{./figs/wnexample2}}
\parpic[r]{\includegraphics[width=2.7cm]{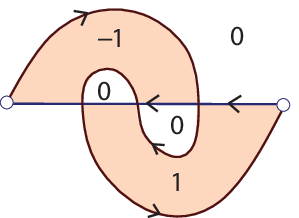} \includegraphics[width=2.5cm]{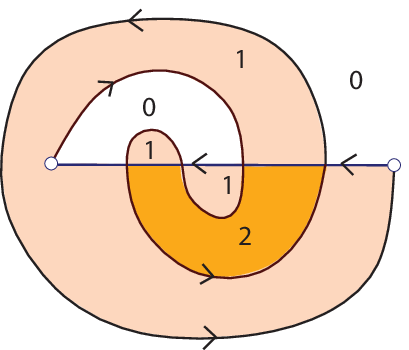}}
We say an oriented curve $\gamma$ has \emph{\allpositive{} winding numbers} if $\wn(x, \gamma)$ is either all non-negative, or all non-positive, for all $x \in \reals^2$. Note that for a curve with \allpositive{} winding numbers, we can always orient the curve appropriately so that $\wn(x, \gamma)$ is all non-negative. Two examples are shown in the figure on the right, where the second example has \allpositive{} winding numbers.  
Let $\arr(\gamma)$ denote the arrangement formed by the curve $\acurve$. All points in the same cell of the arrangement $\arr(\acurve)$ of $\acurve$ have the same winding number, and the winding numbers of two neighboring cells differ by $1$. 
The relation of \allpositive{} winding numbers and \sensepreserving{} homotopies is given below, and the proof can be found in Appendix \ref{appendix:negativewn}. 
\begin{lemma}
If there is a sense-preserving \morph{} $\amorph$ from $P$ to $Q$, then the closed curve $P \concatenate \text{rev}(Q)$ has \allpositive{} winding numbers. 
\label{lem:negativewn}
\end{lemma}

%%%%%%%%%%was in appendix before
\begin{proof}
Without loss of generality, assume that the map $\amorph$ is left \sensepreserving{}, always deforming an intermediate curve to its left. 
Consider the time-varying function $F: [0,1] \times \reals^2 \rightarrow \Z$, where $F(t, x) = \wn(x; \amorph_t)$ is the winding number at $x \in \reals^2$ with respect to the curve parameterized by $\amorph_t$. 
Obviously, $F(0, x) = \wn(x; P \concatenate \text{rev}(Q))$, and $F(1, x) = 0$. 
During the deformation, $F(t, x)$ changes by either $1$ or $-1$ whenever the intermediate curve sweep over it. 
Since the \morph{} is left \sensepreserving{}, when an intermediate curve sweeps $x$, $x$ always moves from the left side of the intermediate curve to its right side. Hence the winding number $x$ decreases monotonically. Since in the end, the winding number at each point is zero, $\wn(x; P \concatenate \text{rev}(Q)) = F(0,x) \ge 0$. 

If the map $\amorph$ is left \sensepreserving{}, then a symmetric argument shows that $\wn(x; P \concatenate \text{rev}(Q)) \le 0$ for all $x \in \reals^2$. 
\end{proof}

Next, we describe two results to connect the above lemma to the computation of optimal \morph{}. First, we define the \emph{\totalW{} $\totalwn(\gamma)$} of a curve $\gamma$ as 
$$\totalwn(\gamma) = \int_{\reals^2} \wn(x; \gamma) d\nu(x),$$ 
where $d\nu(x)$ is the area form\footnote{Note that this allows us to use any Riemannian metric on the plane (including the standard Euclidean metric). This will be necessary later when we use the same algorithm for curves in a universal covering space.}.  
%(Note that this is also known as the algebraic area of the curve; we use total winding number here to emphasize the connection to winding numbers.)  
The following observation is straightforward.

\begin{obs}
For any $P$ and $Q$ in the plane, 
$$\simC(P, Q) \ge | \totalwn(P \concatenate \text{rev}(Q)) |.$$
\label{obs:lowerbound}
\end{obs}
\begin{proof}
Take any regular \morph{} $H$ from $P$ to $Q$. 
%Consider the function $F: [0,1] \rightarrow \reals$ defined as $F(t) = \totalwn(\optmorph_t \concatenate Q)$. Obviously, $F(0) = \totalwn(P\concatenate Q)$, $F(1) = 0$, and $F$ is a continuous function. Furthermore, each time the winding number at a point $x$ changes by 1 for some $t \in [0,1]$, it means that some intermediate curve $\amorph(t)$ sweeps through it. 
The area of a regular homotopy $H$ in our setting can be reformulated as an integral on the image domain as 
$$Area(H) = \int_{\reals^2} deg_H(x) d \nu(x),$$ 
where $deg_H(x)$ is defined as the number of connected components in the pre-image of $x$ under $H$. In other words, $deg_H(x)$ is the number of times that any intermediate curve $H_t$ sweeps through $x$. 
Now consider the function $F: [0,1] \rightarrow \reals$ defined as $F(t) = \totalwn(\optmorph_t \concatenate Q)$. Obviously, $F(0) = \totalwn(P\concatenate Q)$, $F(1) = 0$, and $F$ is a continuous function. Furthermore, each time the winding number at a point $x$ changes by 1 for some $t \in [0,1]$, it means that some intermediate curve $\amorph(t)$ sweeps through it. 
Hence $|\wn(x)|$ is a lower bound for $deg_H(x)$. We thus have that 
$$|\totalwn(P \concatenate \text{rev}(Q))| \le \int_{\reals^2} |\wn(x)| d\nu(x) \le \int_{\reals^2} deg_H(x) d\nu(x)$$ 
for any regular homotopy $H$, implying that 
$$|\totalwn(P\concatenate Q)| \le \inf_{H} Area(H) = \simC(P,Q).$$
%Now, since any optimal homotopy $H$ is tracing an immersed disk in the plane with boundary $P\concatenate Q$, the winding number is a lower bound for the number of times the pre-image of any point $x$ must appear; this is the minimum number of times $x$ must be swept by any continuous deformation between the curves, since it is the number of times the curve is ``wrapped" around the point $x$.  The claim then follows.
%In other words, the winding number at $x$ is the lower bound for the degree of $x$ (which is the number of connected component in the pre-image of $x$ under $\optmorph$: this is how many time $x$ will be swept during the deformation). The claim then follows.
\end{proof}
\begin{figure*}[tbp]
\begin{center}
\begin{tabular}{ccccccc}
\includegraphics[height=3cm]{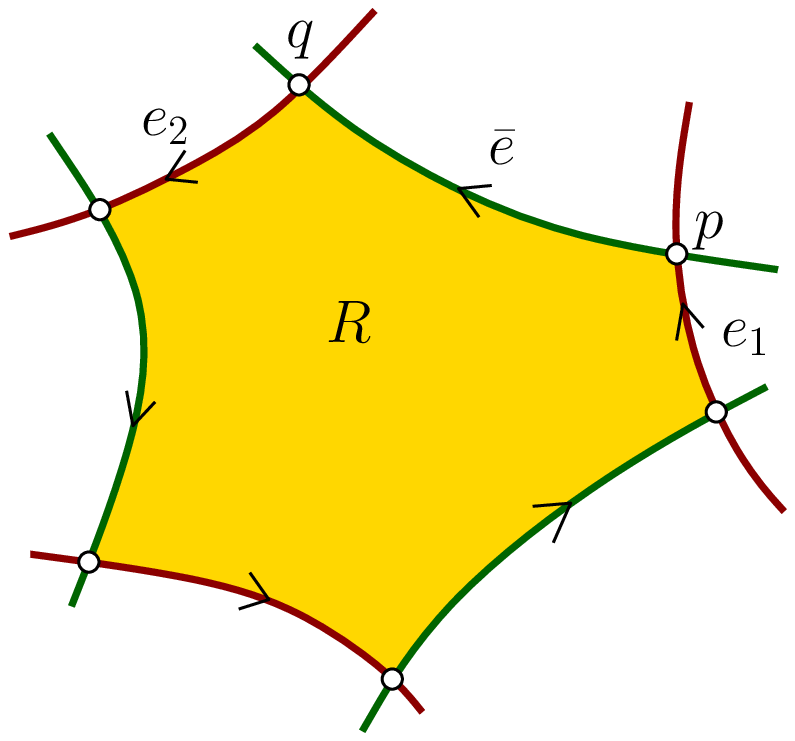} & & 
\includegraphics[height=3cm]{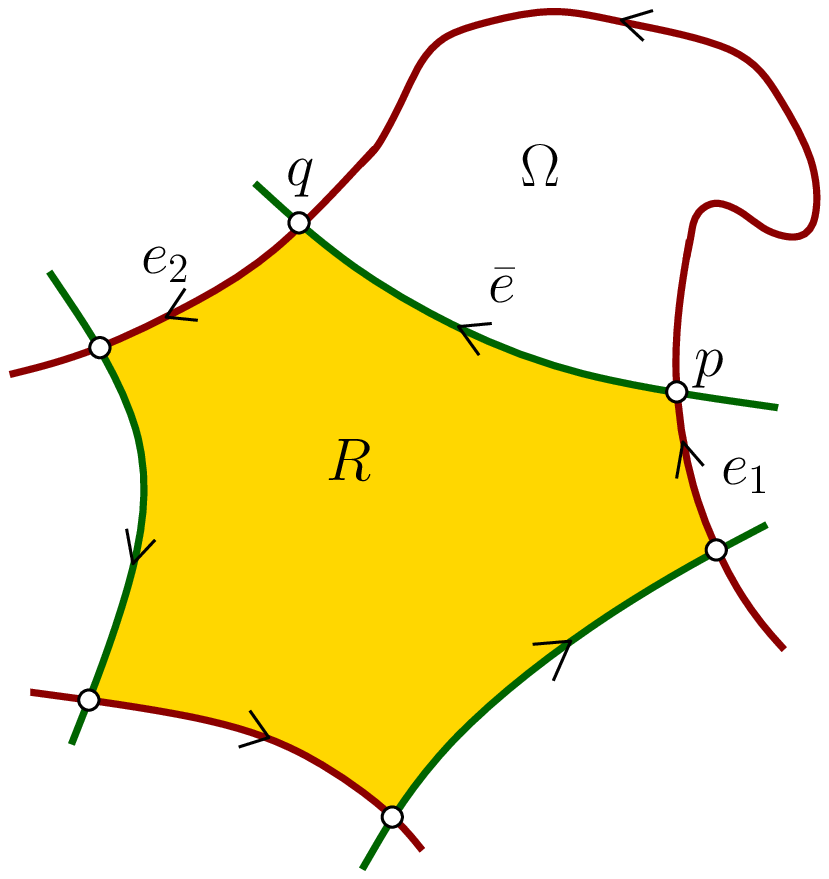} & &
\includegraphics[height=3cm]{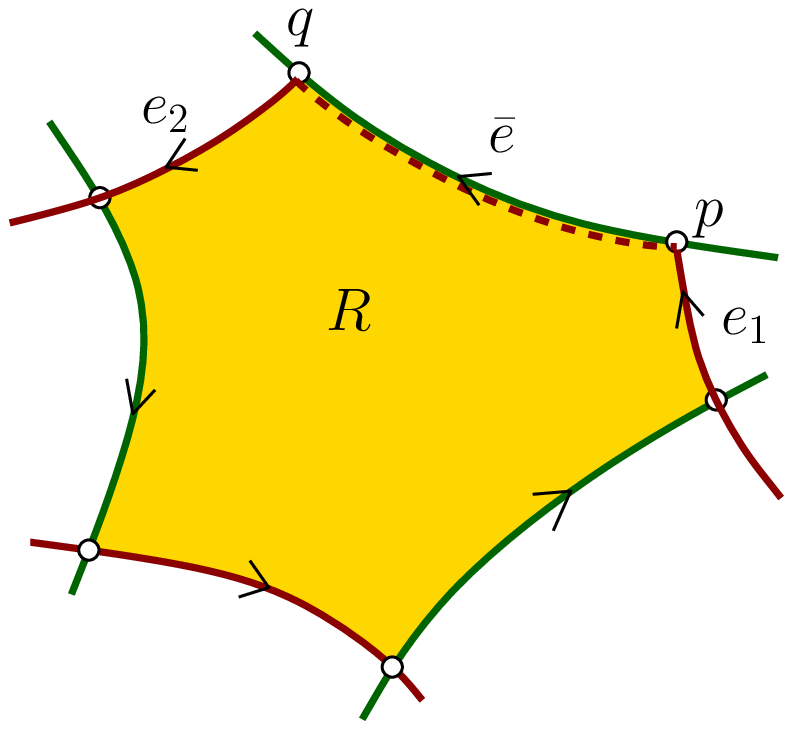} & &
\includegraphics[height=3cm]{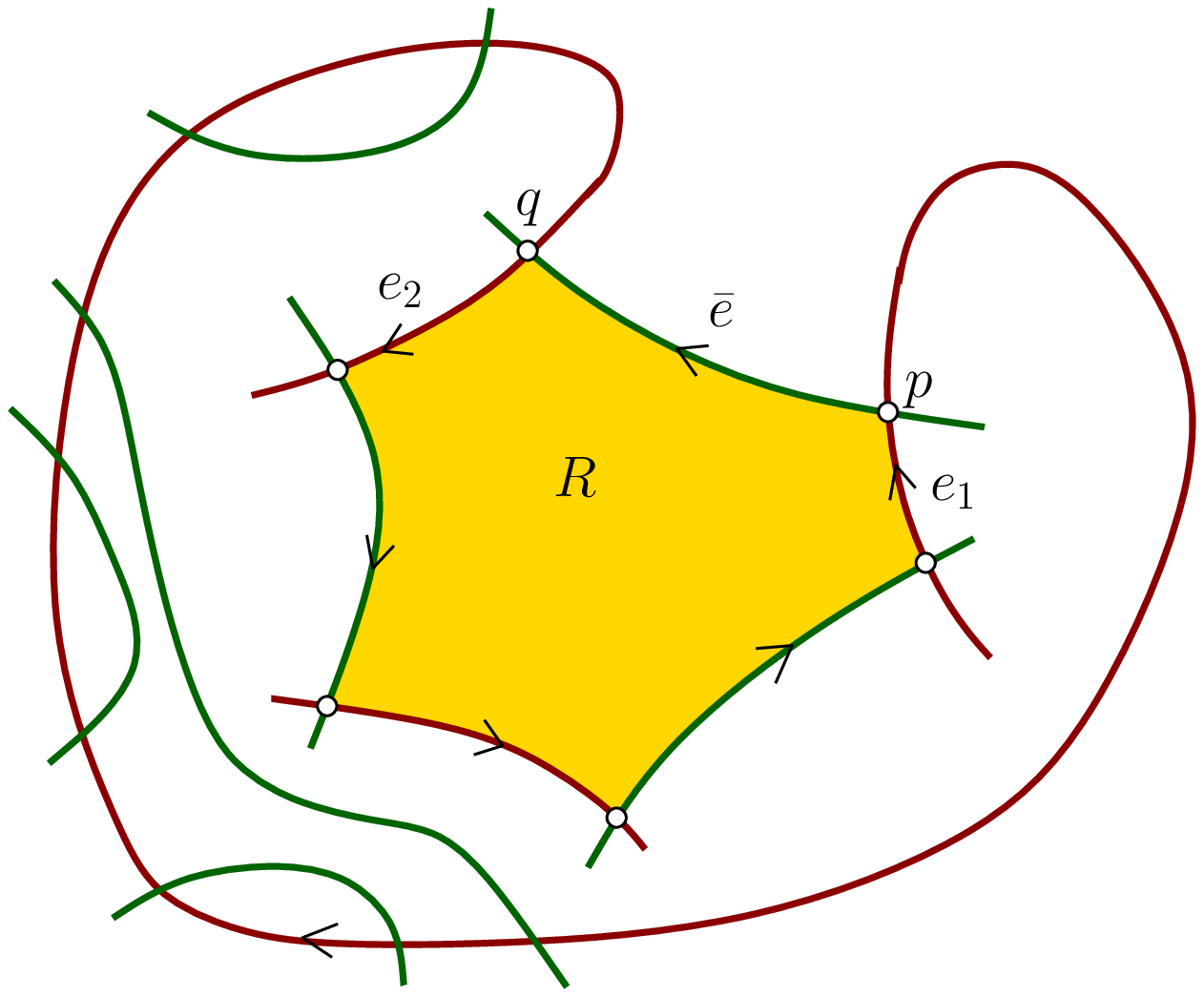} \\ %& &
(a) & & (b) & & (c) & & (d) %& & (e)
\end{tabular}
%\begin{tabular}{ccccc}
%\includegraphics[height=2.5cm]{./figs/maxwncell2} & &
%\includegraphics[height=2.5cm]{./figs/newmaxwncell3} & &
%\includegraphics[height=2.5cm]{./figs/maxwncell5} \\
%(a) & & (b) & & (c) 
%\end{tabular}
\end{center}
\vspace*{-0.15in}
\caption{{\small (a) The cell $\acell$ with highest positive winding number. Its boundary consists of alternating $P$-arcs (red) and $Q$-arcs (green). The two cases of relations between $P[p,q]$ and $\acell$ are shown in (b) and (d), respectively. For case (b), we can deform $P$ to sweep through $\Omega$ as shown in (c), and reduce the number of intersections by $2$. %Similarly, for case (d), we can identify any bigon $R'$ and deform $P$ to reduce the number of intersections by $2$ as well. 
}}
%\caption{{\small Two cases of relations between $P[p,q]$ and $\acell$ are shown in (a) and (b). For case (a), we can deform $P$ to $P'$ as shown in (c), and reduce the number of intersections by $2$. }}
\label{fig:positivewn}
\vspace*{-0.1in}
\end{figure*}
\begin{lemma}
Given $P$ and $Q$, if ~$\acurve = P \concatenate \text{rev}(Q)$ has \allpositive{} winding numbers, then $\simC(P,Q) = | \totalwn(\acurve) |$. 
\label{lem:positivewn}
\end{lemma}

\myproofbegin %\begin{proof}
We  provide a sketch of the proof here to illustrate the main idea; see \cite{fullver} for the full proof. 
We prove the claim by induction on the number of intersections between $P$ and $Q$. 
The base case is when there is no intersection between $P$ and $Q$. 
In this case, $\acurve$ is a Jordan curve which decomposes $\reals^2$ into two regions, one inside $\acurve$ and one outside. The optimal homotopy area $\simC(P,Q)$ in this case is the area of the bounded cell. All points in the bounded cell have winding number $1$ (or $-1$) and the claim follows.  

%\parpic[r]{\includegraphics[width=3cm]{./figs/maxwncell}}
Now assume that the claim holds for cases with at most $k-1$ intersections. We next prove it for the case with $k$ intersections. Let an \emph{$X$-arc} denote a subcurve of curve $X$. 
Consider the arrangement $\arr(\acurve)$ formed by $\acurve = P \concatenate \text{rev}(Q)$. Since $P$ and $Q$ are simple, every cell in this arrangement has boundary edges alternating between $P$-arcs and $Q$-arcs. Assume without loss of generality that $\acurve$ has all non-negative winding numbers. 
Consider a cell $\acell \in \arr(\acurve)$ with largest (and thus positive) winding number. Since its winding number is greater than all its neighbors, it is necessary that all boundary arcs are oriented consistently as shown in Figure \ref{fig:positivewn} (a), where the cell $\acell$ (shaded region) lies to the left of its boundary arcs. 

\parpic[r]{\includegraphics[width=3cm]{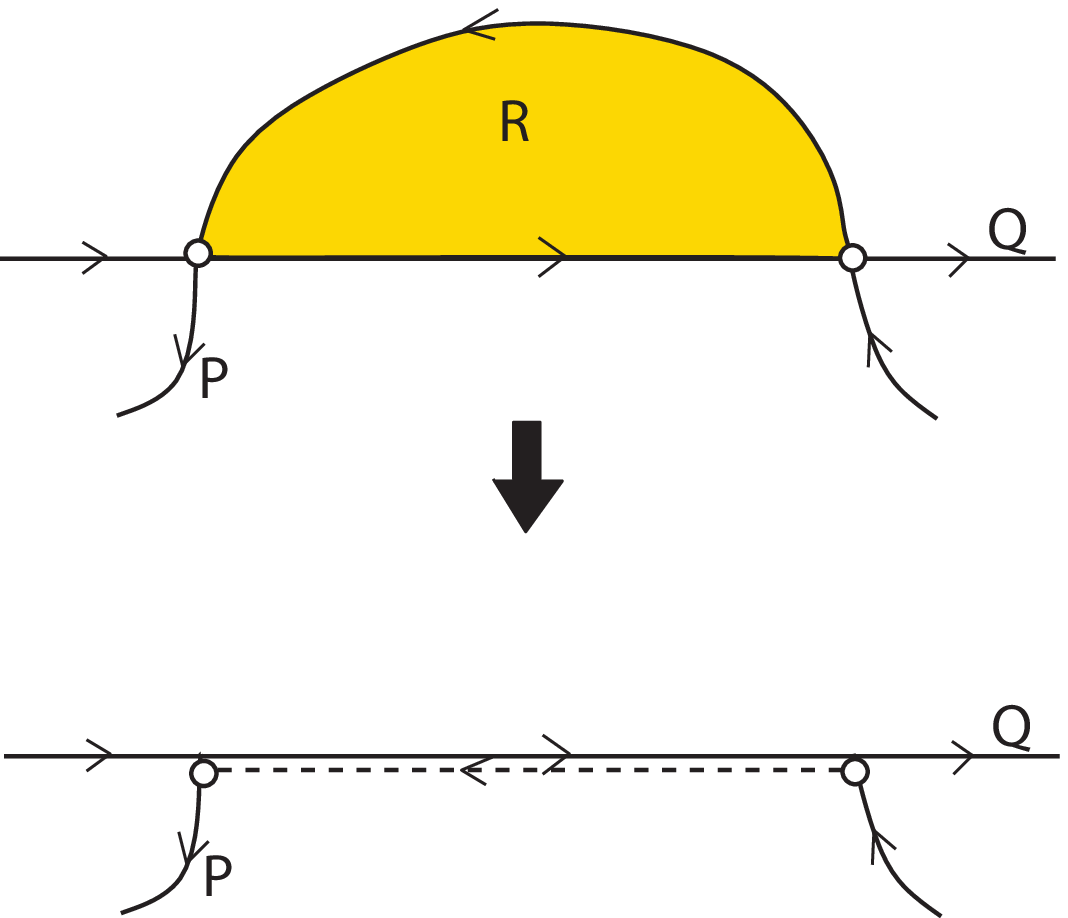}}
If $\acell$ has only two boundary arcs, $e$ from $P$ and $e'$ from $Q$, respectively, then we can morph $P$ to another simple curve $P'$ by deforming $e$ through $\acell$ to $rev(e')$ as illustrated on the right. 
%$-e'$ (where `$-$' means reversing the orientation) as illustrated on the right. 
The area swept by this deformation is exactly the area of cell $\acell$. 
Furthermore, after the deformation, every point $x \in \acell$ decreases their winding number by $1$, and no other point changes its winding number. 
Since any point in this cell initially has strictly positive winding number, the resulting curve $\acurve' = P' \concatenate Q$ still has all non-negative winding numbers. The number of intersections between $P'$ and $Q$ is $k - 2$. By induction hypothesis, $\simC(P', Q) = \totalwn(\acurve')$. Since $\totalwn(\acurve) - \totalwn(\acurve') = \area(\acell)$, we have that $\totalwn(\acurve) = \simC(P', Q) + \area(\acell)$. It then follows from Observation \ref{obs:lowerbound} and the fact $\simC(P, Q) \le \simC(P', Q) + \area(\acell)$ that $\simC(P, Q) = \totalwn(\acurve)$. 

Otherwise, the cell $\acell$ has more than one $P$-arc. Take the $P$-arc $e_1$ with the smallest index along $P$, and let $p$ be its ending point. 
Let $e_2$ be the next $P$-arc along the boundary of $\acell$, and $q$ its starting point, and $Q[q,p]$ the $Q$-arc between $e_1$ and $e_2$, denoted by $\bar{e}$ in Figure \ref{fig:positivewn}. 
$P[p,q]$ and $-Q[p,q]$ bound a simple polygon, which we denote by $\Omega$. Since $\Omega$ does not intersect $R$, 
either $\Omega$ is on the opposite side of the $Q$-arc $\bar{e}$ from the interior of $\acell$ (Figure \ref{fig:positivewn} (b)), 
or they are on the same side (Figure \ref{fig:positivewn} (d)). 
It turns out that in both cases, we can deform $P$ to another simple curve $P'$ such that (i) the number of intersections is reduced by $2$, and (ii) $P' \concatenate Q$ still has consistent winding numbers. For example, in the case of Figure \ref{fig:positivewn} (b), $P$ is then deformed to sweep the region $\Omega$ as shown in Figure \ref{fig:positivewn} (c). By applying the induction hypothesis to $P' \concatenate Q$, we are able to obtain that $\simC(P, Q) = \totalwn(\acurve)$. 
 The details are in Appendix \ref{appendix:positivewn}. 
\myproofend

%%%%%%%%%%%%%%%%%%%%%%%%%%%%
\subsection{The Algorithm} 
\label{subsec:algorithm}

%Again, we focus on the case where $P$ and $Q$ are simple curves (instead of cycles) sharing the same endpoints. 
Lemma \ref{lem:sensepreserving} and \ref{lem:negativewn} imply that if the closed curve $P \concatenate \text{rev}(Q)$ produces both positive and negative winding numbers, then any optimal \morph{{} from $P$ to $Q$ must have at least one \breakpt{}. On the other hand, if it has \allpositive{} winding numbers, then by Lemma \ref{lem:positivewn} we can compute the optimal \energy{} to deform them by simply computing the total winding number. This leads to a simple dynamic-programming (DP) approach to compute $\simC(P,Q)$. 

Specifically, let $\Ipt_0, \Ipt_1, \ldots, \Ipt_I$ denote the intersection points between $P$ and $Q$, ordered by their indices along $P$, where $\Ipt_0$ and $\Ipt_I$ are the beginning and ending points of $P$ and $Q$, respectively. 
Let $T(i)$ be the \energy{} of the optimal \morph{} between $P[\zerobf, \Ipt_i]$ and $Q[\zerobf, \Ipt_i]$, and $C[i,j]$ the closed curve formed by $P[\Ipt_i, \Ipt_j] \concatenate Q[\Ipt_j, \Ipt_i]$. We say that a pair of indices $(i,j)$ is \emph{valid} if (1) $\Ipt_i$ and $\Ipt_j$ have the same order along $P$ and along $Q$; and (2) the closed curve $C[i,j]$ has \allpositive winding numbers. We have the following recursion: 
\begin{eqnarray*} 
T(i) = 
\begin{cases} 
  \totalwn(C[0, i]),  & \mbox{If }C[0, i] \mbox{ has \allpositive{} winding numbers;} \\
  \min_{j < i, \mbox{~$(j, i)$ is valid}} ~\{ ~\totalwn(C[j, i]) + T(j) ~\}, & \mbox{ Otherwise.} 
\end{cases}
\end{eqnarray*}

%\myparagraph{Time complexity. }
\subsection{Time Complexity Analysis}
\label{subsec:time}

The main components of the DP framework described above are (i) to compute $\totalwn(C[i,j])$ for all pairs of $i,j$s, and (ii) to check whether each pair $(i,j)$ is valid or not. These can be done in $O(I^2n)$ total time in a straightforward manner. We now show how to compute them in $O(I^2\log I)$ time after $O(I\log I + n\log n)$ pre-processing time. 
Specifically, we describe how to compute such information in $O(I\log I)$ time for all $C[r,i]$s for a fixed $r \in [1, I]$ and all indices $i > r$. 

To simplify the description of the algorithm, we extend $Q$ on both sides until infinity, and obtain $\extendQ$. Now collect all intersection points between $P$ and $\extendQ$, $\{ \extIpt_1, \ldots, \extIpt_I \}$, which is a super-set of previous intersection points, and sort them by their order along the curve $P$. 
The algorithm can be made to work with $Q$ directly, but using $\extendQ$ makes the intuition behind our algorithm, as well as its description, much more clear. 

Note that $\extendQ$ divides the plane into two half-planes. For illustration purpose, we will draw $\extendQ$ as a horizontal line, and use the upper and lower half-planes to refer to these two sides of $\extendQ$. Another way to see that regarding $\extendQ$ as a horizontal line does not cause any loss of generality is that one can always find a homeomorphism from $\reals^2 \rightarrow \reals^2$ such that the image of $\extendQ$ is a horizontal line under this homeomorphism. 

Now for a fixed integer $r$, we traverse $P$ starting from $\extIpt_r$. We aim to maintain appropriate data structures so that each time we pass through an intersection point $\extIpt_i$ with $\extendQ$, we can, in $O(\log I)$ time, (1) check whether $(r, i)$ is valid, and (2) obtain total winding number for $C[r,i]$. 

%\parpic[r]{\includegraphics[width=5.5cm]{./figs/Rus1}}
\begin{figure}[htbp]
\centerline{\includegraphics[width=5.5cm]{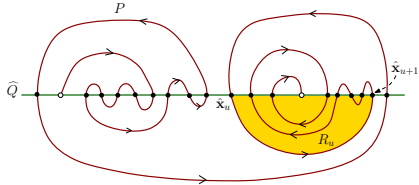}}
\caption{Illustration of the regions $R_u$s.}
\label{fig:Rus1}
\end{figure}
\myparagraph{Total winding numbers.}
We first explain how to maintain the total winding number for the closed curve $C[r,i] = P[\extIpt_r, \extIpt_i] \concatenate Q[\extIpt_i, \extIpt_r]$ as $i$ increases. 
Assume $i$ changes from $u$ to $u+1$. Since $\extIpt_u$ and $\extIpt_{u+1}$ are two consecutive intersection points along $P$, the arcs $P[\extIpt_u, \extIpt_{u+1}]$ and $\extendQ[\extIpt_u, \extIpt_{u+1}]$ form a simple closed polygon which we denote by $R_u$ (shaded region Figure \ref{fig:Rus1}). Comparing the arrangement $\arr(C[r,u+1])$ with $\arr(C[r, u])$, regardless of where $r$ is, only points within $R_u$ will change their winding number, either all by $+1$ or all by $-1$, depending on whether $R_u$ is to the right side or the left side of the $P$-arc $P[\extIpt_u, \extIpt_{u+1}]$, respectively. The winding numbers for points outside $R_u$ are not affected. 
Hence the change in the total winding number is simply $\alpha_u \area(R_u)$, where $\alpha_u$ is either $+1$ or $-1$. See Figure \ref{fig:Rus1}, where all points in $R_u$ will decrease their winding number by $1$ as we move from $C[r,u]$ to $C[r,u+1]$. 

We can pre-compute the area of $R_u$'s for all $u$ in $O(n\log n+I\log I)$ time, by observing that the set of $R_u$s satisfy the parentheses property: Namely, either $R_u$ and $R_v$ are disjoint in their interior, or one contains the other. 

\begin{figure}[h]
\centerline{\includegraphics[width=4.5cm]{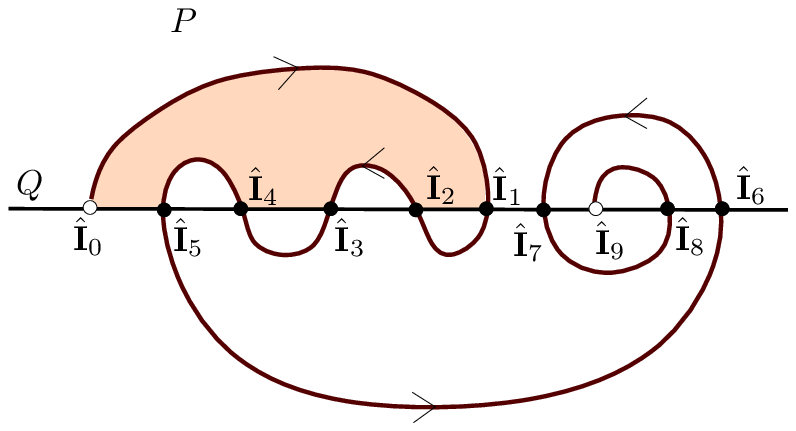} \hspace*{0.25in} \includegraphics[width=5cm]{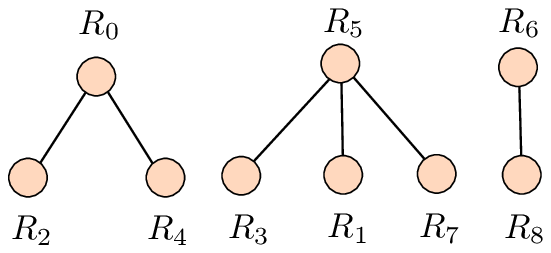}}
\caption{The containment relations of all $R_u$ regions can be represented as a forest structure on the right. }
\label{fig:containmenttree}
\end{figure}
%\parpic[r]{\includegraphics[width=4cm]{./figs/nextIpt2} \hspace*{0.05in} \includegraphics[width=4cm]{./figs/nextIpttree}}
%\myparagraph{Computing the area of $R_u$ for all $u \in [1, I]$. }
Specifically, first, we compute the arrangement of $\arr(P+\extendQ)$ and the area of all cells in it in $O(n\log n +I)$ time.  
Each $R_u$ is the region bounded between a $P$-arc $P[\extIpt_i, \extIpt_{i+1}]$ and a corresponding $\extendQ$-segment $\extendQ[\extIpt_i, \extIpt_{i+1}]$. 
Since no two $P$-arcs intersect, the containment relationship between such $P$-arcs satisfies parentheses property. In particular, we can use a collection of trees to represent the containment relation among all regions $R_u$s. See Figure \ref{fig:containmenttree} for an illustration. The difference between the region represented at a parent node and the union of regions represented by all its children is a cell in $\arr(P+Q)$. For example, the shaded cell in the right figure is the difference between $R_0$ and its children $R_2$ and $R_4$. 
We can thus compute the area of all $R_u$s by a bottom-up traversal of these trees. Computing these trees take $O(I \log I)$ time by first sorting all intersection points with respect to their order along $\extendQ$. Traversing these trees to compute all $R_u$s takes $O(I)$ time. Putting everything together, we need $O(n\log n + I\log I)$ time. 

With the area of all $R_u$s known, updating the total winding number from $C[r,u]$ to $C[r, u+1]$ takes only constant time. 

\myparagraph{Checking the validity of $(r,i)$s. }
To check whether $(r, i)$ is valid or not, we need to check whether all cells in the arrangement $\arr(C[r,i])$ have consistent winding numbers. 
First observe that for any $r$ and $i$, $\arr(P + \extendQ)$ is a refinement of the arrangement $\arr(C[r,i])$. That is, a cell in $\arr(P + \extendQ)$ is always contained within some cell in $\arr(C[r,i])$. Hence all points within the same cell of $\arr(P+\extendQ)$ always have the same winding number with respect to any $C[r,i])$, and we simply need one point from each cell in $\arr(P+\extendQ)$ to maintain the winding number for all cells in $\arr(C[r,i])$, for any $r$ and $i$. 
We now describe how to maintain the winding number for cells of $\arr(P+\extendQ)$ (thus for $\arr(C[r,u])$s) as we pass each $u > r$, so that we can check whether $C[r,u]$ has \allpositive{} winding numbers or not efficiently. 

\parpic[r]{\includegraphics[width=4cm]{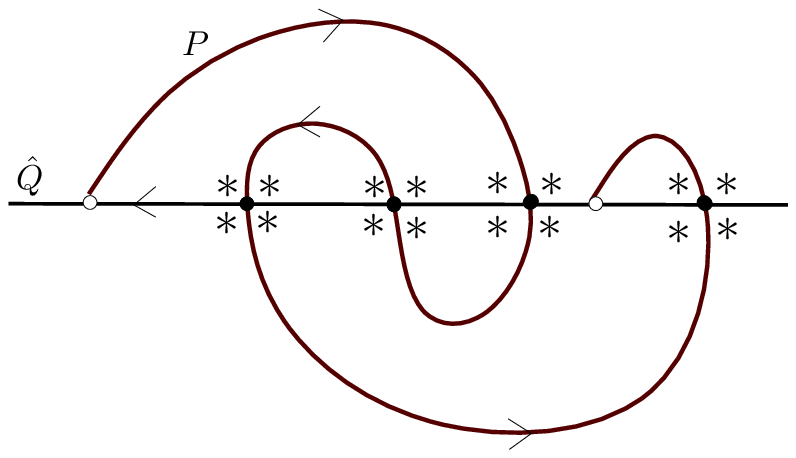}}
\noindent To this end, take four points around each intersection point $\extIpt_i$ of $P$ and $\extendQ$ (shown as stars in the right figure). The collection of such \emph{\representative{} points} hit all cells in $\arr(P+\extendQ)$. (It does not matter whether there may be more than one point taken from a cell of $\arr(P + \extendQ)$.) Hence $\arr(C[r,i])$ has consistent winding number if and only if all these \representative{} points have consistent winding numbers. 
Next, we build a data structure to maintain the winding numbers for these points as $i$ increases. 
Specifically, let $U$ be the set of \representative{}s that are to the right of $\extendQ$, which are the stars above $\extendQ$ in the right figure. (Those to the left of $\extendQ$ will be handled in a symmetric manner). Each point has a key associated with it which is its index along $\extendQ$. We build a standard balanced $1$-D range tree on $U$ based on such keys, where each leaf $f$ stores a point from $U$. Every internal node $v$ is associated with an interval $[l_v, r_v]$, where $l_v$ and $r_v$ are the smallest and largest keys stored in the subtree rooted at $v$. In other words, all \representative{}s with an index along $\extendQ$ within $[l_v, r_v]$ are stored in the subtree rooted at $v$. 
At every node $v$, interior or not, we also store a value $\addwn_v$. To compute the winding number for the \representative{} point $p_f$ stored at a leaf node $f$, we identify the path $\{ v_0, v_1, \ldots, v_a = f \}$ from the root $v_0$ to $f$. The winding number for $p_f$ is simply $\sum_{i=0}^a \addwn_{v_i}$. 
Finally, each internal node $v$ also stores the maximum and minimum winding numbers associated with all leaves in its subtree. At the beginning, all winding numbers are zero. The size of this tree is $O(I)$ with height $O(\log I)$, and can be built in $O(I\log I)$ time once the arrangement $\arr(P + \extendQ)$ is known. 

Let $\qin_i$ denote the index of point $\extIpt_i$ along $\extendQ$ (or can be considered as the $x$-coordinate of $\extIpt_i$). 
As we move from $C[r,u]$ to $C[r,u+1]$, cells of $\arr(P+\extendQ)$ contained in $R_u$ should either all increase or all decrease their winding number by $1$. Note that \representative{}s of these cells are simply those contained in the interval $[\qin_u, \qin_{u+1}]$ (or $[\qin_{u+1},\qin_u]$ if $\qin_{u+1} < \qin_u$). Hence updating the winding number is similar to an interval query of $[\qin_u, \qin_{u+1}]$, and the $O(\log I)$ number of nodes in the canonical decomposition of $[\qin_u, \qin_{u+1}]$ update their $\addwn_v$ values by either $+1$ or $-1$ depending on the sideness of $R_u$ with respect to the arc $P[\extIpt_u, \extIpt_{u+1}]$. 
The minimum and maximum winding numbers can also be updated $O(1)$ time per visited node. The entire process visits $O(\log I)$ nodes, and thus takes $O(\log I)$ time as $i$ increases from $u$ to $u+1$. 
To see whether $C[r, u+1]$ has \allpositive{} winding numbers or not, we only need to check the minimum and maximum winding numbers stored at the root of the tree, denoted by $w_{min}$ and $w_{max}$, respectively. If $w_{min} \times w_{max}$ equals to zero, then all winding numbers w.r.t. $C[r, u+1]$ are either all non-negative or all non-positive. Otherwise, $(r,u+1)$ is not valid. 

Repeat the above process for every $r \in [1, I]$. Overall, after $O( (n+I)\log n)$  pre-processing, we can check whether $(r, i)$ is valid or not and compute $\totalwn(C[r,i])$ for all $r \in [1, I]$ and all $i > r$ in $O(I^2\log I)$ time. 
Putting everything together, we have the following result. 
%our dynamic programming algorithm takes $O(I^2\log I + n\log n)$ time and $O(I^2+n)$ space to compute $\simC(P,Q) = T[I]$. 

\begin{theorem}
Given two simple polygonal chains $P$ and $Q$ (with the same endpoints) in the plane of $n$ total complexity, and with $I$ intersection points between them, we can compute the optimal \morph{} and its area in $O(I^2 \log I + n\log n)$ time and $O(I^2 + n)$ space. 
\label{thm:plane}
\end{theorem}

The case where we have two simple cycles $P$ and $Q$ in the plane is discussed in Appendix \ref{appendix:sec:planarcycles}, and we obtain the following extension: 

\begin{corollary}
Given two polygonal cycles $P$ and $Q$ in the plane of $n$ total complexity and with $I$ intersection points, we can compute the optimal \morph{} and its area in $O(I(I^2 \log I + n \log n))$ time if $I > 0$; and compute the optimal homotopy area in $O(n\log n)$ time if $I = 0$.  
\label{cor:planecycles}
\end{corollary}

\section{Minimum Area Homotopies on 2-Manifolds}
\label{sec:surface}

In this section, we consider optimal \morph{} between curves $P$ and $Q$ on an orientable and triangulated $2$-manifold $M$ without boundary. Our input is a triangulation $\trimesh$ of $M$ with complexity $\tricomplexity$, and two simple homotopic polygonal curves $P$ and $Q$ sharing endpoints. Edges in $P$ and $Q$ are necessarily edges from the triangulation $\trimesh$. The total complexity of $P$ and $Q$ is $n$, and there are $I$ number of intersections between them. Note that in this setting, $I = O(n)$. 
Below we discuss separately the cases when $M$ has non-zero genus
% and the input is either two curves or two cycles, as well as the case 
and when $M$ is a topological sphere. 
%The extension of our algorithm to cycles on surfaces can be found in Appendix \ref{appendix:cycles-surface}, which is more involved than the extension in the planar case. 

%%%%%%%%%%%%%%%%%%%%%%%%%
\subsection{Surfaces with Positive Genus}
\label{sec:positivegenus}

Given an orientable $2$-manifold $M$, let $\uc(M)$ be a universal covering space of $M$ with $\phi: \uc(M) \rightarrow M$ the associated covering map. Note that $\phi$ is continuous, surjective, and a local homeomorphism.  (For full details on covering spaces, we refer the reader to topology textbooks that address this area~\cite{Rot88}; we will also build on existing algorithmic techniques developed  for the computing and working in the universal cover~\cite{DS95,Sch92}.)

For any path $\gamma$ in $M$, if we fix the lift (pre-image) of its starting point, then it lifts to a unique path $\lift \gamma$ in $\uc(M)$, such that $\phi(\lift \gamma) = \gamma$. 
Since $P$ and $Q$ are homotopic with common endpoints, the closed curve formed by $C = P\concatenate Q$ is contractible on $M$, and the lift of $C$, denoted by $\lift C$, is a closed curve in $\uc(M)$. By the Homotopy Lifting Property of the universal cover \cite{Rot88}, we have:

\begin{obs}
Once we fix the lift of the starting point of $P$ and $Q$ in $\uc(M)$, there is a one-to-one correspondence between homotopies between $P$ and $Q$ in $M$ and those between $\lift P$ and $\lift Q$ in $\uc(M)$. 
\label{obs:lifting}
\end{obs}

We now impose an area measure in $\uc(M)$ by lifting the area measure in $M$; this can be done via the map $\phi$, which is a local homeomorphism. 
Now the area of a \morph{} in $M$ is the same as the area of its lift in $\uc(M)$. As such, we can convert the problem of finding an optimal \morph{} in $M$ to finding one in $\uc(M)$. 
Furthermore, for any orientable compact 2-manifold with genus $g > 0$, its universal cover is topologically equivalent to $\reals^2$. Intuitively, this means that we can then apply algorithms and results from previous section to the universal covering space. 

%The entire algorithm takes $O(I^2 \log I + ng\log n + \tricomplexity)$ time where $g$ is the genus of the surface $M$. 
%Roughly speaking, we construct the portion of the universal covering space that the lift of $C$ will traverse (and enclose in some sense), which consists of $O(n)$ copies of some polygonal schema of $M$ \cite{DS95}. The main observation is that only the combinatorial structure of $\lift{C}$ is needed, so we can avoid filling in each copy of the polygonal schema with all triangles from $\trimesh$. Details are in Appendix \ref{appendix:algnonzerogenus}. 

%We now describe our algorithm to compute the optimal homotopy for two homotopic polygonal chains $P$ and $Q$ from a triangulation $\trimesh$ (with complexity $\tricomplexity$) of a compact orientable surface $M$ with genus $g > 0$. There are $n$ total edges in $P$ and $Q$, each being an edge of $\trimesh$, and there are $I$ intersections between $P$ and $Q$.  %obviously, $I = O(n)$ in this scenario. 
%%
%The high level idea is to compute a lift of the curve $C = P \concatenate \text{rev}(Q)$ in the universal covering space, and then apply the algorithm from Section \ref{sec:plane} to the lifted curves. We now detail the steps involved and their complexity. \\

More specifically, our algorithm proceeds as follows: \\
\\
%\myparagraph{Step 1: Compute relevant portion of a universal covering space.}
\noindent {\bf Step 1: Compute relevant portion of $\uc(M)$.~}
We will construct a portion of a universal covering space $\uc(M)$ made from polygonal schema of $M$~\cite{VY90,DS95}. Specifically, we use the algorithm from~\cite{DS95} to construct a reduced polygonal schema $T$ in $O(\tricomplexity)$ time. The universal covering space consists of an infinite number of copies of this polygonal schema glued together appropriately. We call each copy of the polygonal schema in the constructed universal covering space \emph{a tile}. 

Recall that the universal covering space $\uc(M)$ is homeomorphic to $\reals^2$. We fix a lift of the starting endpoint of $P$ and $Q$ in $\uc(M)$ and obtain a specific lift $\lift{P}$ and $\lift{Q}$ for $P$ and $Q$ respectively. Since $P$ and $Q$ are homotopic, $\lift{P}$ and $\lift{Q}$ form a closed curve, denoted by $\lift{C} = \lift{P} \concatenate \text{rev}(\lift{Q})$. 
Note that the number of intersection points between $\lift{P}$ and $\lift{Q}$ is at most $I$, as every intersection point in the lift necessarily maps to an intersection point of $P$ and $Q$ under $\phi$, but not vice versa. 

Consider the arrangement formed by $\arr(\lift{C})$ in the planar domain $\uc(M)$.  We will construct the portion of the universal covering space $U \subseteq \uc(M)$ which is the union of tiles that intersect or are contained inside of $\arr(\lift{C})$. 

From~\cite{DS95}, we know that the lifted curve $\lift{C}$ passes through $O(n)$ tiles in $\uc(\trimesh)$. However, while the total number of tiles in the interior of $\arr(\lift{C})$ is $O(n)$ for the case where $g > 1$, it can be $\Theta(n^2)$ for the case when $g = 1$. Hence we will separate the case for $g = 1$ and $g > 1$, since we wish to avoid the $O(n^2)$ overhead in the genus 1 case.  

For the case $g > 1$, we use the algorithm by Dey and Schipper \cite{DS95} to compute the relevant portion $U$ of the universal covering space in $O(n\log g + \tricomplexity)$ time. The output contains all $O(n)$ copies of the polygonal schema in $U$, where each tile is represented by a reduced $4g$-gon without being explicitly filled with triangles from $\trimesh$. 
However, recall that $P$ and $Q$ are curves which follow edges of the triangulation; in this construction of the polygonal schema tiles, each edge of $K$ can be broken into $O(g)$ pieces. So in the worst case, we must break each edge in $P$ or $Q$  into $O(g)$ pieces, giving a total  complexity for $\tilde{P}$ and $\tilde{Q}$ is $O(ng)$ in this representation of $U$.   
Once these are known, we can compute the combinatorial structure of the arrangement of $\lift{C}$ in $U$, as well as the description of the set of tiles each  cell in $\arr(\lift{C})$ intersects or contains, in $O(ng + I\log I)$ time. 

\begin{figure}[htbp]
\begin{center}
\begin{tabular}{ccc}
\includegraphics[width=3.8cm]{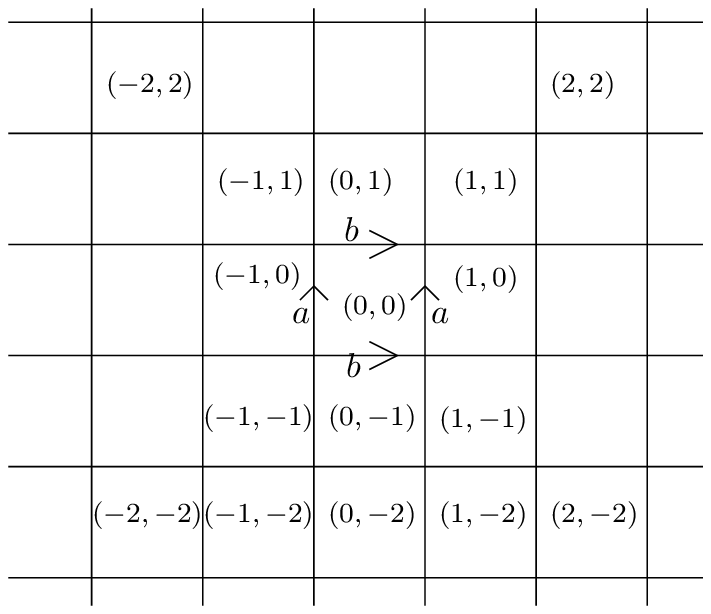} & \hspace*{0.0in} &
\includegraphics[width=4cm]{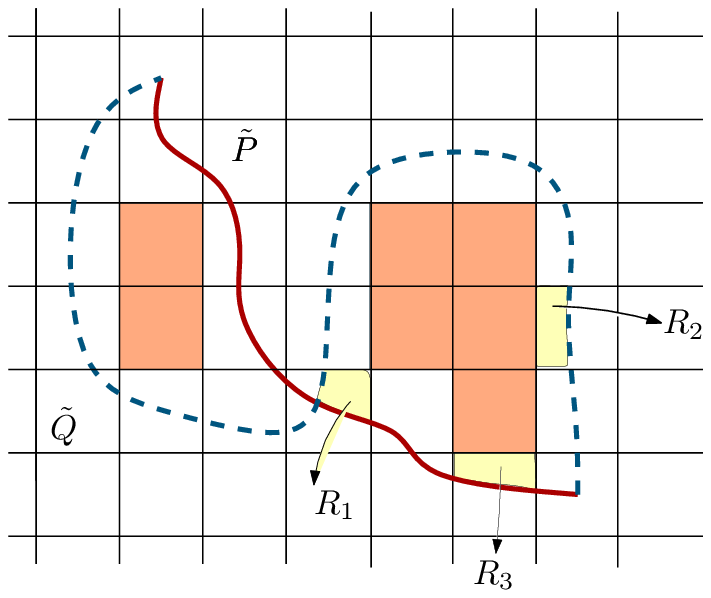} \\
(a) & & (b)
\end{tabular}
\end{center}
\vspace*{-0.15in}
\caption{{\small (a) A combinatorial view of the universal covering space $\uc(M)$. $a$ and $b$ are the generators and we can give each cell a coordinate. (b) The lift of $P$ (solid curve) and the lift of $Q$ (dashed curve). The heavily shaded region are copies of polygonal schema contained inside cells of $\arr(\lift{C})$, and their total number can be easily computed by a scanning algorithm. $R_1$ is an essential cell; $R_2$ and $R_3$ are two non-essential cells. }
\label{fig:torusUC}}
\end{figure}

For the case $g = 1$,  the input manifold is a torus, and the canonical polygon schema for it is a rectangle with oriented boundary arcs $aba^{-1}b^{-1}$. Imagine now that we give the \emph{base polygonal schema} $T_0$ (which is the tile that contains the lift of the starting point of $P$ and $Q$) a \emph{coordinate} $(0,0)$; we must now assign a coordinate for every other copy of the polygon (as shown in Figure~\ref{fig:torusUC}(a)). Specifically, a copy of polygonal schema $T$ has coordinate $(i,j)$ if the closed loops whose lifts start in $T_0$ and end in $T$ have the same homotopy type as $a^i b^j$. 
We can  obtain the sequence of the rectangles (and their coordinates) that the curve $\lift C$ will pass through in $O(n + \tricomplexity)$ time \cite{DS95}. Once these coordinates are known, the combinatorial structure of the arrangement of $\lift{C}$ in $U$ can also be computed in $O(n + I\log I)$ time. Note that in this case, we have avoided explicitly enumerating the set of tiles fully enclosed within $\arr(\lift{C})$ (the shaded tiles in Figure \ref{fig:torusUC} (b)), whose number can be $\Theta(n^2)$ instead of $O(n)$ when $g=1$.

\myparagraph{Step 2: Area of cells in $\arr(\lift{C})$.}
In order to perform our algorithm introduced in Section \ref{sec:plane} to the lifted curves $\lift{P}$ and $\lift{Q}$, in addition to the combinatorial structure of $\arr(\lift{C})$, we also need the area of each cell in $\arr(\lift C)$. We first describe how to compute it for the case $g = 1$. 

Take any cell $X$ in $\arr(\lift C)$ and assume the boundary of $X$ intersects $m$ copies of polygonal schema. 
Even though that $X$ may contain $\Theta(m^2)$ copies of (rectangular) tiles in its interior, 
we do not need to enumerate these interior tiles explicitly to compute their total area. 

Indeed, by a scanning algorithm from left to right, we can compute in $O(m)$ time how many tiles are completely contained inside $X$ (heavily--shaded regions in Figure \ref{fig:torusUC} (b)) (note that the coordinates of each tile traversed by the boundary of $X$ are known). 
Since the area of every polygonal schema is simply the total area of the input triangulation, we can compute the total area of tiles contained inside $X$ in $O(m)$ time. 

Now let $\mathbf{R}$ be the collection of tiles that intersect the boundary of $X$. It remains to compute the total area of $\mathbf{R} \cap X$. 
Call each region in $T \cap X$ a \emph{sub-cell}, for any tile $T \in \mathbf{R}$. 
Let $G$ denote the boundary curves of the polygonal schema $T$. 
There are two types of sub-cells: the \emph{essential} ones which contain at least one intersection point between $\lift{P}$ and $\lift{Q}$ as their vertices, and the \emph{non-essential} ones which have no intersection; see Figure~\ref{fig:torusUC} (b) for examples. Note that a non-essential cell is bounded  by arcs from $G$ alternating with $P/Q$-arcs from $\lift{P}$ or $\lift{Q}$, since there is no intersection of $\lift{P}$ and $\lift{Q}$ along the boundary of a non-essential cell. (Here, a $P/Q$-arc refers to either a $P$-arc or a $Q$-arc). 

%\begin{wrapfigure}{r}{3.5in}
\begin{figure}[tbhp]
\begin{center}
\begin{tabular}{cc}
\includegraphics[height=5cm]{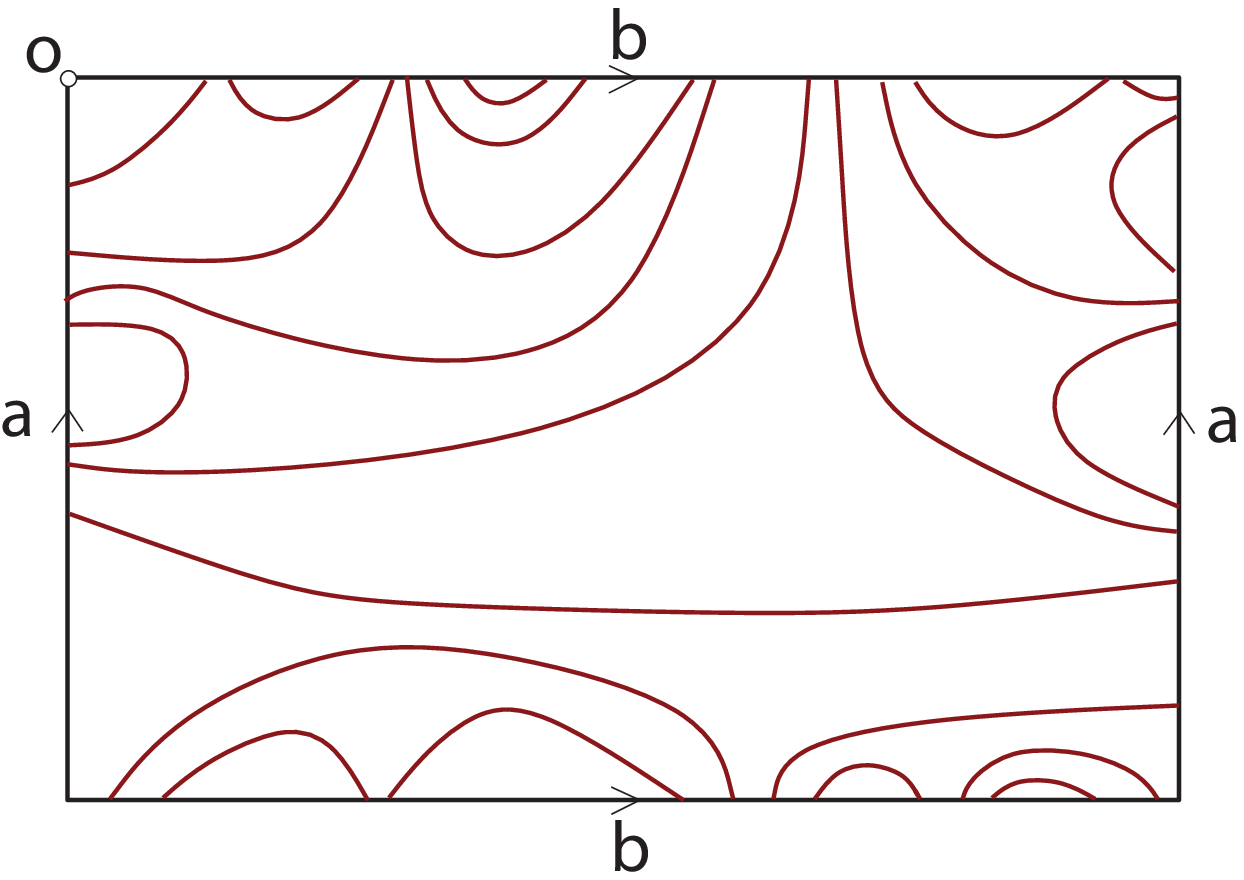}  &
\includegraphics[height=5cm]{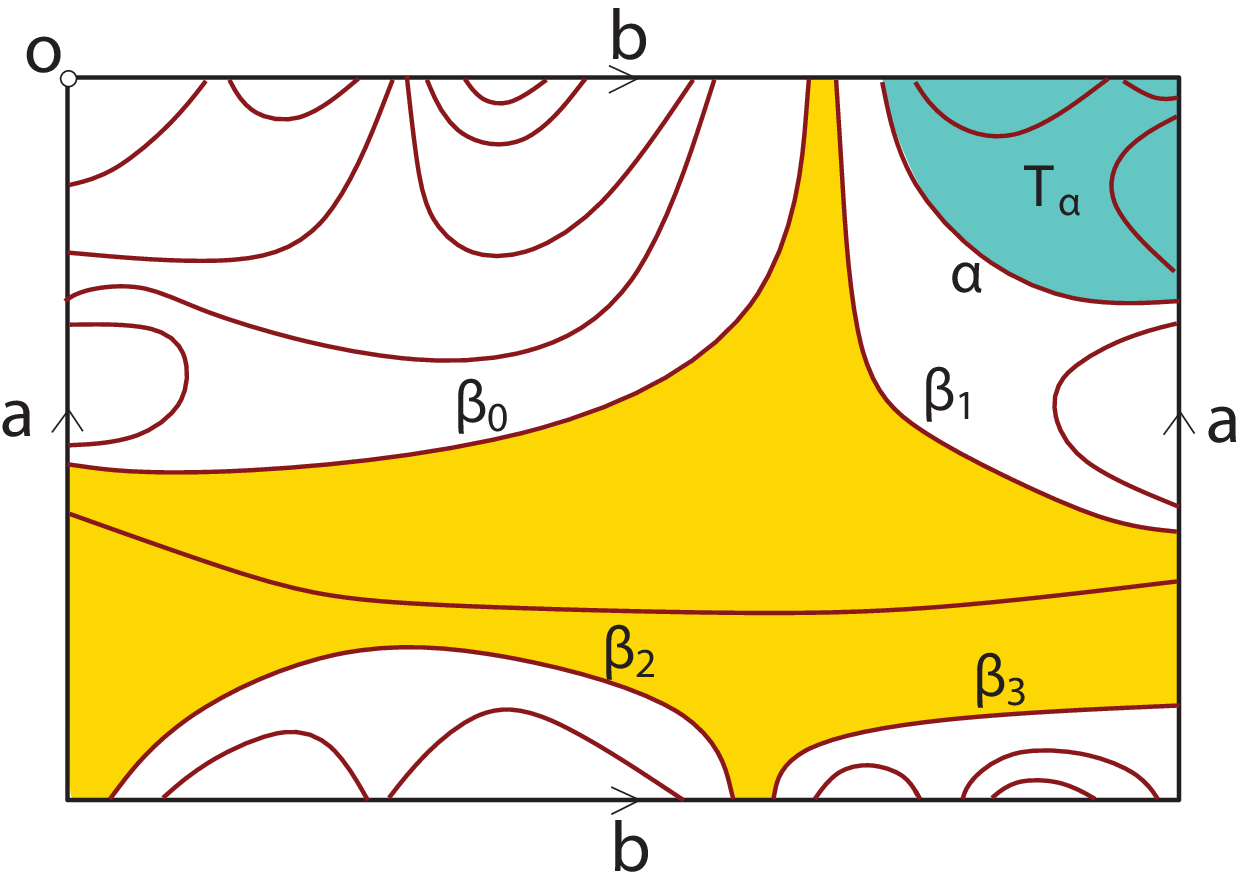} \\
(a) &  (b)
\end{tabular}
\end{center}
\vspace*{-0.15in}
\caption{(a) We overlay all non-essential sub-cells involving $P$-arcs into one copy of the polygonal schema. (b) An example of the canonical region $T_\alpha$ is shown for arc $\alpha$ (shaded region in the top-right corner). The shaded region in the middle is a sub-cell $X$ which can be computed as $T_{\beta_0} - T_{\beta_1} - T_{\beta_2} - T_{\beta_3}$, where $\beta_i$s are the boundary $P$- and $Q$-arcs for $X$. Among these $P$/$Q$-arcs, $\beta_0$ is the top-most arc in the containment relation. 
}
\label{fig:nonessential}
\end{figure}
%\end{wrapfigure}
First let us consider the collection of non-essential sub-cells formed by alternating $G$-arcs (boundary arcs of a tile) and arcs from $\tilde{P}$ and $\tilde{Q}$, and compute the area of each such non-essential sub-cells. 
%If we plot all these $P$-arcs and $Q$-arcs within one tile, no two $P$-arcs and $Q$-arcs in this tile intersect, since $P$ and $Q$ are simple curves and this cell is assumed to be non-essential. 
If we plot all the $P$-arcs within a single tile $T$, no two $P$-arcs can intersect in this tile, since $P$ is a simple curve. 
Imagine that we pick an arbitrary but fixed point on the boundary $G$ of the polygonal schema $T$ as the origin $\mathbf{o}$. Each $P$-arc $\alpha$ subdivides $T$ into two regions; we let $T_\alpha$ denotes the canonical one excluding $\mathbf{o}$. 
Note that since $P$ is a simple curve, the set of \emph{canonical regions} $T_\alpha$s for all $P$-arcs must satisfy the parenthesis property, and these regions and their areas, called \emph{canonical areas}, can be computed in $O(ng\log n + \tricomplexity)$ time using a data structure similar to one used in Section \ref{subsec:time} to compute the area of $R_u$s. 
See Figure \ref{fig:nonessential} for an illustration. 
Similar, we can put all $Q$-arcs within the same tile and compute the canonical regions / areas associated with all $Q$-arcs in $O(ng\log n + \tricomplexity)$ time. 
Once these areas are known, the area of each non-essential sub-cell can be computed in $O(s)$ time where $s$ is the number of $P$-arcs and $Q$-arcs on the boundary of this sub-cell: Specifically, it is the difference between the canonical area of the top-most $P$/$Q$-arc and the union of the canonical areas of all other $P$/$Q$-arcs on the boundary of this cell. 
See Figure \ref{fig:nonessential} (b). 
Hence the areas of all non-essential sub-cells can be computed in $O(n)$ time once all $T_\alpha$s are known. The total time complexity required here is thus $O(ng \log n + \tricomplexity)$. 

What remains is to compute the area of all essential sub-cells. Note that there are $O(I)$ essential sub-cells since each contains an intersection between $P$ and $Q$. 
%Again, let $G$ denote the boundary of the polygonal schema $T$. 
Let a $PQ$-arc to refer to an arc that starts and ends with points on $G$ (the boundary of the polygonal schema $T$) and consists of alternating $P$- and $Q$-arcs. 
An essential sub-cell is either completely contained within a polygonal schema, or its boundary consists of $PQ$-arcs, $G$-arcs, $P$-arcs and $Q$-arcs where no two such arcs can be consecutive: they are separated by $G$-arcs. 
Now collect all $P$-arcs and $Q$-arcs that are involved in the boundary arcs of those essential sub-cells completely contained within a tile. 
Plot them within the same tile $T$ and compute their arrangement $A$ as well as the area for each cell in $A$. This can be done in $O(ng\log n + \tricomplexity)$ time.
Since $A$ can have only $O(I)$ vertices in the interior of the tile $T$, $A$ contains $O(I)$ cells.  
If an essential sub-cell $X$ is completely contained within a polygonal schema, then it is a union of a set of cells from $A$. We can simply spend $O(I)$ time to go through cells in $A$, identify those contained in $X$ and return their total area. 
Hence it takes $O(I^2)$ time to compute the area of all such $O(I)$ essential sub-cells. 
If an essential sub-cell $X$ has $G$-arcs on its boundary, then we need a slightly more complicated way to handle it. 

Specifically, for all the remaining essential sub-cells, there can be $O(I)$ number of $PQ$-arcs along their boundaries, denoted by $L$. 
%(Note that two $PQ$-arcs can share common $P$-arcs or $Q$-arcs among them). 
We collect all $P$-arcs and $Q$-arcs involved in $L$ and plot them in the same tile $T$ and compute their arrangement $\arr(L)$. 
Each $PQ$-arc $\alpha \in L$ divides the tile $T$ into two regions, and we define $T_\alpha$ to be the canonical one that excluding a specific origin $\mathbf{o}$ on $G$ similar to before. 
$T_\alpha$ consists of a union of cells from the arrangement $\arr(L)$, and we can compute the area of $T_\alpha$ in $O(I)$ time since $\arr(L)$ has $O(I)$ cells. 
Overall, in $O(I^2)$ time, we can compute the area of all $T_\alpha$s for all $PQ$-arcs $\alpha \in L$ 
Now take an essential sub-cell $X$ that has $s$ number of $P$-, $Q$-, or $PQ$-arcs along its boundary, denoted by $\alpha_1, \alpha_2, \ldots, \alpha_s$. Let $\alpha_1$ be the arc (which can be $P$-, $Q$- or $PQ$-arc) whose endpoints along $G$ spans the largest interval. Then, $X$ can be represented as $X = T_{\alpha_1} - \bigcup_{i\in [2,s]} T_{\alpha_i}$, where $T_{\alpha_i}$ is the canonical region defined by an arc $\alpha_i$. Since the area of all canonical regions are known (for $P$-arcs or $Q$-arcs, we have computed their canonical areas before), $X$'s are  a can be computed in $O(s)$ time. 
Computing the area of all remaining essential sub-cells thus takes $O(I^2 + n)$ time. 

Putting everything together, the total time needed to compute the area of all cells in $\arr(\lift{C})$ is $O(ng\log n +  \tricomplexity + I^2)$ when $g = 1$. 
The case when $g > 1$ is similar but much simpler. Indeed, we now can afford to compute all the tiles contained within any cell of $\arr(\lift{C})$ explicitly, as their total number is bounded by $O(n)$ \cite{DS95,Sch92}. The areas of essential and non-essential sub-cells are computed using the same algorithm as above. The total time complexity is $O(ng\log n + \tricomplexity + I^2)$. 

\myparagraph{Step 3: Putting everything together.}
With the combinatorial structure of $\arr(\lift{C})$ and the area of each cell computed, we now apply the algorithm from Section \ref{subsec:algorithm} to compute the optimal \morph{} in $O(I^2 \log I + ng \log n)$ time in $\uc(M)$, which, by \ref{obs:lifting}, gives the optimal \morph{} between $P$ and $Q$ in $M$ in the same time bound. The total time complexity for the entire algorithm is $O(ng\log n + I^2\log I + \tricomplexity)$.

%%%%%%%%%%%%%%%%%%%%%%%%%%%%%
\subsection{The Case of the Sphere }
\label{sec:sphere}

We now consider the remaining case where the input has $g=1$, so $M$ is a (topological) sphere $\sphere$. All paths on $\sphere$ are  homotopic. The universal cover of a sphere is itself, and hence is compact. However, the previous algorithm in Section \ref{subsec:algorithm} works for a domain homeomorphic to $\reals^2$ and cannot be directly applied. We now sketch how we handle the sphere case. 
Missing details can be found in Appendix \ref{appendix:sphere}. 
%; see Appendix \ref{appendix:sphere} for the full discussion. 

\begin{figure}[htbp]
\centerline{\includegraphics[width=3cm]{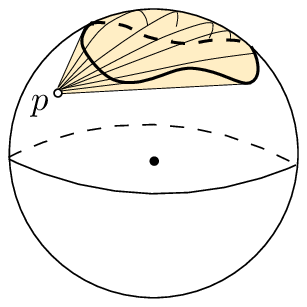} \hspace*{0.35in} \includegraphics[width=3cm]{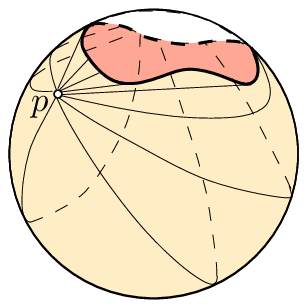}}
\caption{Two ways of sweeping a curve on sphere from base point $p$.}
\label{fig:spherewn}
\end{figure}
%\parpic[r]{\includegraphics[width=3cm]{./figs/sphere1} \hspace*{0.05in} \includegraphics[width=3cm]{./figs/sphere2}}
%For simplicity, assume the input $M$ is the unit sphere $\sphere$. 
We observe that the results in Section \ref{sec:optimal} still hold. 
However, as the sphere is compact, the winding number is not well-defined. For example, see Figure \ref{fig:spherewn}, where there are two ways that  the curve $\gamma$ winds around the point $p$. In the first case, the winding number at $p$ is $0$, while in the second case, the winding number is $-1$. In order to use a dynamic programming framework as before to compute the optimal \morph{} between $P$ and $Q$, we need to develop analogs of Lemma \ref{lem:negativewn} and \ref{lem:positivewn} for curves on the sphere. 

To this end, note that if we remove one point, say $\z\in \sphere$ from the sphere $\sphere$, then the resulting space $\sphereminus{\z} = \sphere - \z$ is homeomorphic to $\reals^2$, and the concept of the winding number is well defined for $\sphereminus{\z}$. Specifically, $\z$ can be considered as the point of infinity in $\reals^2$. The \emph{winding number of $x \in \sphereminus{\z}$ w.r.t. $C$ and $z$}, denoted by $\wn(x; \z, C)$ ($C$ omitted when its choice is clear), is simply the summation of signed crossing numbers for any path connecting $x$ to $\z$. %
%\footnote{This definition is consistent with the topological interpretation that the winding number is the homotopy class of the complementary space.}. 
As in the planar case, we say that $C$ is \emph{\allpositive{} w.r.t. $\z$} if $\wn(x; \z, C)$ is either non-negative, or non-positive for all $x \in \sphereminus{\z}$.
Similar to before, we define the total winding number w.r.t. a base point $\z$ as $\totalwn(C; \z) = \int_{\sphereminus{\z}} \wn(x; \z, C) dx$. 
Let $\simC(P, Q; \Omega)$ denote the best \energy{} to morph $P$ to $Q$ within domain $\Omega$.

\begin{obs}
If there is an optimal \morph{} between $P$ and $Q$ that does not sweep through some point $\z \in \sphere$, then we have $\simC(P, Q; \sphere) = \simC(P, Q; \sphereminus{\z})$. 
\label{obs:pointofinfinity}
\end{obs}

\begin{obs}
Suppose $\optmorph$ is an optimal \morph{} between $P$ and $Q$ with no \breakpt{}s. 
For any cell $\acell$ in $\arr(P+Q)$, if $\optmorph$ sweeps through one point in its interior, then it sweeps through all points in $\acell$. 
\label{obs:onecell}
\end{obs}

%\begin{obs}
%Given a closed curve $\Gamma$ and any two points $\z, \w \in \sphere$, we have that: $\wn(x; \w) = \wn(x; \z) + \wn(\z; \w)$ (all winding numbers are w.r.t the curve $\Gamma$). 
%In particular, for any two points $\z_1, \z_2$ from the same cell of $\arr(\Gamma)$, we have that $\wn(x; \z_1) = \wn(x; \z_2)$ for all $x \neq \z_1, \z_2$. 
%\label{obs:samecellwn}  
%\end{obs}

The simple proof for the above observation is in Appendix \ref{appendix:obs:onecell}. The key result is the following lemma, the proof of which can be found in Appendix \ref{appendix:lem:infpointexists}. 
\begin{lemma}
If there is an optimal \morph{} $\optmorph$ of $P$ and $Q$ with no \breakpt{}, then the image of this optimal homotopy cannot cover all points in $\sphere$. 
\label{lem:infpointexists}
\end{lemma}

%These results have the following implications. 
Given two homotopic paths $P$ and $Q$ from $\sphere$ sharing common endpoints, 
Lemma \ref{lem:infpointexists} and Observation \ref{obs:pointofinfinity} imply 
that if $P$ can be morphed to $Q$ optimally without \breakpt{}s, 
then there exists some point $\z \in \sphere$ such that $\simC(P,Q; \sphere) = \simC(P,Q; \sphereminus{\z})$. 
For this choice of $\z$, it is necessary that the closed curve $P\concatenate Q$ has consistent winding numbers in $\sphereminus{\z}$. 
Once this $\z$ is identified, $\simC(P,Q; \sphereminus{\z})$ is simply the total winding number of $P\concatenate Q$ w.r.t. $\z$, as suggested by Lemma \ref{lem:positivewn}, because $\sphereminus{\z}$ is homeomorphic to the plane. 
Furthermore, by Observation \ref{obs:onecell}, we only need to pick one point from each cell of $\arr(P+Q)$ to check for the potential $\z$. 
Specifically, let $\{ \z_1, \ldots, \z_l \}$ be a set of such \emph{\representative{}s}, where $l = O(I)$. 
The optimal homotopy area $\simC(P, Q)$ is simply the smallest of all $\totalwn(P \concatenate \text{rev}(Q); \z_i)$ %
%$\simC(P', Q'; \z_i)$ 
for those $\z_i$s with respect to whom the curve $P \concatenate \text{rev}(Q)$ has \allpositive{} winding numbers. 
Hence if we assume that if there is an \emph{optimal} \morph{} between $P$ and $Q$ with \emph{no} \breakpt{}s, then we can compute $\simC(P,Q)$. 
%then the results in this section provide an algorithm to compute $\simC(P,Q)$. 
%This will form the base case to handle the general case as we discuss below. 

\myparagraph{Overview of the algorithm for sphere case. }
To compute the optimal \morph{} between $P$ and $Q$, we follow the same dynamic programming framework as before. If there is no \breakpt{} in an optimal \morph{}, then we use the discussion above to compute the optimal homotopy area. Otherwise, we identify the intersection point that serves as next \breakpt{}, and recurse. 
The main difference lies in the component of computing $\simC(i,j) := \simC(P[\Ipt_i, \Ipt_j], Q[\Ipt_i, \Ipt_j])$, assuming that there is an optimal homotopy from $P' = P[\Ipt_i, \Ipt_j]$ to $Q' = Q[\Ipt_i, \Ipt_j]$ with no \breakpt{}s. 
Previously, this is done by checking whether $P'\concatenate Q'$ has \allpositive{} winding numbers. Now, we need to check the same condition but against $l = O(I)$ number of potential \representative{}s $\{ \z_1, \ldots, \z_l \}$ as the potential point of infinity. This gives a linear-factor blow-up in the time complexity compared to the algorithm for the planar case. However, we show that this linear blow-up can be tamed down and we can again compute \emph{all} $\simC(r, j)$s for all $r$s and all $j > r$ in $O(I^2\log I)$ time, after $O(n \log n + \tricomplexity)$ pre-processing time. See Appendix \ref{appendix:algsphere} for details. 
%A straightforward implementation of this checking, just for \emph{one} pair $(i,j)$, takes $O(I^2n)$ time, by computing each $\simC(P', Q'; \sphereminus{\z_k})$, for $k \in [1, l]$, in $O(n)$ time independently. 
%This can be improved to $O(I^2\log I)$ time to compute \emph{all} $\simC(r, j)$s for a fixed $r$ and all $j > r$, after $O(n \log n + \tricomplexity)$ pre-processing time. See Appendix \ref{appendix:algsphere} for details. 
Overall, the total time complexity remains the same as before. 
%, which is $O(I^2\log I + n \log n  +\tricomplexity)$, where $\tricomplexity$ is total complexity of the triangulation for $\sphere$.  

Putting both cases ($g > 0$ and $g = 0$) together, we conclude with the following main result. 

\begin{theorem}
Given a triangulation $\trimesh$ of an orientable compact $2$-manifold $M$ with genus $g$, let $\tricomplexity$ be the complexity of $\trimesh$. Given two homotopic paths $P$ and $Q$ of total complexity $n$ with $I$ intersection points, we can compute an optimal homotopy and its area $\simC(P,Q; M)$ in $O(I^2\log I + ng\log n + \tricomplexity)$ time. 

\label{thm:surface}
\end{theorem}

%%%%%%%%%%%%%%%%%%%%%%%%%%%%%%
\section{Conclusion}
% homology case
% immersed disk
In this paper, we propose a new curve similarity measure which captures how hard it is to deform from one curve to the other based on the amount of total area swept. It is robust to noise (as it is area-based), and can be computed  efficiently; to our knowledge, there is no other efficiently computable similarity measure for curves on surfaces. 
Our algorithm can be extended for cycles in the plane (see Appendix \ref{appendix:sec:planarcycles}).  
It appears that our algorithm can also be extended to cycles on the surfaces.  Indeed, if the optimal free homotopy has an anchor point, then we can break cycles into curves that share a common start and end point, which then reduces to the problem of comparing curves on surfaces.  However, on a surface the analog of Lemma \ref{lem:anchorcycles} no longer holds, so that two curves may intersect in $M$ but not have an anchor point in the optimal homotopy; in this case, it is not immediately clear how to bound the size of the universal cover necessary for our algorithm.  We leave the problem of working out these details, as well as improving the time complexity for comparing cycles, as an immediate future direction.

%The algorithm can also be extended to compare simple homotopic cycles on surfaces, although the time complexity is larger. An immediate next question is whether we can improve the time complexity of our algorithm for cycles to remove the extra linear or quadratic factor that we currently need. 

Currently, we assume that two input paths are simple paths which share starting and ending points, which makes it easier to define homotopy equivalence. This leads to two natural questions, namely how to handle curves which do not share endpoints and how to deal with non-simple curves.  Another interesting problem is to compute optimal \emph{isotopy area} where we require that any intermediate curve during the deformation is also simple. 

Measuring similarity of curves on surfaces is an interesting problem, and many open areas remain. 
%Current methods can usually be viewed as computing certain minimum deformation cost. 
Geodesic \Frechet{}-based measures ignore the topological constraints of underlying surface, while the homotopic \Frechet{} distance, homotopy height, and our method require identification of a homotopy which optimizes some cost. 
As far as other measures of similarity which may be tractable, one interesting new idea would be to develop an area-based curve similarity measure that allows topological changes, such as allowing a region to be swept as long as it has trivial homology.
Other directions include developing efficient curve simplification algorithms based on this measure, and studying similarity between curves from more general simplicial complexes than considered in this paper (such as a manifold with boundary or holes, or non-manifolds).

%%%%%%%%%%%%%%%%%%%%%%%%%%%%%%%%%%%%%%%%%%%%%%%%%%%%%
\paragraph{Acknowledgment.}
The authors would like to thank Joseph O'Rourke, Rephael Wenger,  Michael Davis, and Tadeusz Januszkiewicz for useful discussions at the early stage of this work, and David Letscher, Brody Johnson, and Bryan Clair for helpful discussions at the later stage of this paper. We would also like to thank the anonymous reviewers for their comments. This work is partially supported by the National Science Foundation under grants CCF-0747082 and CCF-1054779. 

%\newpage
%%%%%%%%%%% bibliography %%%%%%%%%%%%%%%%%
%\bibliographystyle{salpha}%
%\bibliographystyle{alpha}%
\bibliographystyle{abbrv}
%\bibliography{shortcuts,part_f}
\bibliography{geometry,curvesim,shortcuts}

\providecommand{\Badoiu}{B\u{a}doiu} \providecommand{\Matousek}{Matou{\v s}ek}
  \providecommand{\Barany}{B{\'a}r{\'a}ny}
  \providecommand{\Bronimman}{Br{\"o}nnimann}
  \providecommand{\Gartner}{G{\"a}rtner} \providecommand{\Badoiu}{B\u{a}doiu}
  \providecommand{\tildegen}{{\protect\raisebox{-0.1cm}
  {\symbol{'176}\hspace{-0.03cm}}}} \providecommand{\SarielWWWPapersAddr}
  {http://www.uiuc.edu/~sariel/papers} \providecommand{\SarielWWWPapers}
  {http://www.uiuc.edu/\tildegen{}sariel/\hspace{0pt}papers}
  \providecommand{\urlSarielPaper}[1]{ \href{\SarielWWWPapersAddr/#1}
  {\SarielWWWPapers{}/#1}}
\begin{thebibliography}{10}

\bibitem{AAKS13}
P.~K. Agarwal, R.~B. Avraham, H.~Kaplan, and M.~Sharir.
\newblock Computing the discrete {Fr\'{e}chet} distance in subquadratic time.
\newblock In {\em Proc. 24th Ann. ACM-SIAM Sympos. Discrete Alg. (SODA)}, 2013.

\bibitem{alt2009}
H.~Alt.
\newblock The computational geometry of comparing shapes.
\newblock In {\em Efficient Algorithms}, volume 5760 of {\em Lecture Notes in
  Computer Science}, pages 235--248. Springer Berlin / Heidelberg, 2009.

\bibitem{afrw-mcsrs-98}
H.~Alt, U.~Fuchs, G.~Rote, and G.~Weber.
\newblock Matching convex shapes with respect to the symmetric difference.
\newblock {\em Algorithmica}, 21:89--103, 1998.

\bibitem{ag-cfdbt-95}
H.~Alt and M.~Godau.
\newblock Computing the {Fr\'echet} distance between two polygonal curves.
\newblock {\em International Journal of Computational Geometry and its
  Applications}, 5:75--91, 1995.

\bibitem{ag-dgsmi-00}
H.~Alt and L.~J. Guibas.
\newblock Discrete geometric shapes: {Matching}, interpolation, and
  approximation.
\newblock In J.-R. Sack and J.~Urrutia, editors, {\em Handbook of Computational
  Geometry}, pages 121--153. Elsevier Science Publishers B. V. North-Holland,
  Amsterdam, 2000.

\bibitem{akw-cdmpc-04}
H.~Alt, C.~Knauer, and C.~Wenk.
\newblock Comparison of distance measures for planar curves.
\newblock {\em Algorithmica}, 38(1):45--58, 2004.

\bibitem{ahkww-fdcr-06}
B.~Aronov, S.~{Har-Peled}, C.~Knauer, Y.~Wang, and C.~Wenk.
\newblock {Fr\'{e}chet} distance for curves, {R}evisited.
\newblock In {\em Pro. of European Symposium on Algorithms}, pages 52--63,
  2006.

\bibitem{b-ombp-02}
S.~Bespamyatnikh.
\newblock An optimal morphing between polylines.
\newblock {\em International Journal of Computational Geometry and
  Applications}, 12(3):217--228, 2002.

\bibitem{BC06}
P.~Bose, S.~Cabello, O.~Cheong, J.~Gudmundsson, M.~var Kreveld, and
  B.~Speckmann.
\newblock Area-preserving approximations of polygonal paths.
\newblock {\em Journal of Discrete Algorithms}, 4:554--566, 2006.

\bibitem{Brightwell2009}
G.~R. Brightwell and P.~Winkler.
\newblock Submodular percolation.
\newblock {\em SIAM J. Discret. Math.}, 23(3):1149--1178, July 2009.

\bibitem{bbkrw-hdwd-07}
K.~Buchin, M.~Buchin, C.~Knauer, G.~Rote, and C.~Wenk.
\newblock How difficult is it to walk the dog?
\newblock In {\em Proc. 23rd Europ. Workshop Comput. Geom.}, pages 170--173,
  2007.

\bibitem{bbmm-fswd-12}
K.~Buchin, M.~Buchin, W.~Meulemans, and W.~Mulzer.
\newblock Four {S}oviets walk the dog---with an application to {A}lt's
  conjecture.
\newblock {\em arXiv/1209.4403}, 2012.

\bibitem{bbklsww-mt-10}
K.~Buchin, M.~Buchin, M.~van Kreveld, M.~L\"offler, R.~I. Silveira, C.~Wenk,
  and L.~Wiratma.
\newblock Median trajectories.
\newblock In {\em Proc. 18th Annual European Symposium on Algorithms (ESA)},
  volume 6346, pages 463--474. Springer, 2010.

\bibitem{BBW09}
K.~Buchin, M.~Buchin, and Y.~Wang.
\newblock Partial curve matching under the fr\'{e}chet{} distance.
\newblock In {\em Proc. ACM-SIAM Sympos. Discrete Alg. (SODA)}, 2009.

\bibitem{Buchin07}
M.~Buchin.
\newblock {\em On the Computability of the {Fr\'{e}chet} Distance Between
  Triangulated Surfaces}.
\newblock {PhD} thesis, Dept. of Comput. Sci., Freie Universit\"{a}t Berlin,
  2007.

\bibitem{CVE08}
E.~W. Chambers, {\'E}.~{Colin de Verdi{\`e}re}, J.~Erickson, S.~Lazard,
  F.~Lazarus, and S.~Thite.
\newblock Homotopic fr̩chet distance between curves or, walking your dog in
  the woods in polynomial time.
\newblock {\em Computational Geometry}, 43(3):295 -- 311, 2010.
\newblock Special Issue on 24th Annual Symposium on Computational Geometry
  (SoCG'08).

\bibitem{Chambers_isotopicfrechet}
E.~W. Chambers, T.~Ju, D.~Letscher, and L.~Liu.
\newblock Isotopic fr\'{e}chet distance.
\newblock In {\em CCCG11}, 2011.

\bibitem{chambersletscher2009}
E.~W. Chambers and D.~Letscher.
\newblock On the height of a homotopy.
\newblock In {\em CCCG'09}, pages 103--106, 2009.

\bibitem{CW08}
A.~F. Cook and C.~Wenk.
\newblock Geodesic {Fr\'{e}chet} distance inside a simple polygon.
\newblock In {\em Proc. 25th Internat. Sympos. Theoret. Asp. Comp. Sci.}, pages
  193--204, 2008.

\bibitem{Cro92}
R.~G. Cromley.
\newblock {\em Digital Cartography}.
\newblock Prentice Hall, Englewood Cliffs, NJ, 1992.

\bibitem{DS95}
T.~K. Dey and H.~Schipper.
\newblock A new technique to compute polygonal schema for $2$-manifolds with
  application to null-homotopy detection.
\newblock {\em Discrete and Computational Geometry}, 14(1):93--110, 1995.

\bibitem{douglas}
J.~Douglas.
\newblock Solution of the problem of plateau.
\newblock {\em Trans. of the American Mathematical Society}, 33:263--321, 1931.

\bibitem{dreimelharpeled2012}
A.~Driemel and S.~Har-Peled.
\newblock Jaywalking your dog: computing the fr\'{e}chet distance with
  shortcuts.
\newblock In {\em Proc. 23rd Annual ACM-SIAM Symposium on Discrete Algorithms},
  SODA '12, pages 318--337. SIAM, 2012.

\bibitem{EFV07}
A.~Efrat, Q.~Fan, and S.~Venkatasubramanian.
\newblock Curve matching, time warping, and light fields, new algorithms for
  computing similarity between curves.
\newblock {\em J. Math. Imaging Vis.}, 27(3):203--216, 2007.

\bibitem{ehgmm-nsmpa-02}
A.~Efrat, S.~{Har-Peled}, L.~J. Guibas, J.~S. Mitchell, and T.~Murali.
\newblock New similarity measures between polylines with applications to
  morphing and polygon sweeping.
\newblock {\em Discrete and Computational Geometry}, 28:535--569, 2002.

\bibitem{nomagicleash2012}
S.~Har-Peled, A.~Nayyeri, M.~Salavatipour, and A.~Sidiropoulos.
\newblock How to walk your dog in the mountains with no magic leash.
\newblock In {\em Proc. 28th Symposuim on Computational Geometry}, SoCG '12,
  pages 121--130, New York, NY, USA, 2012. ACM.

\bibitem{harpeledraichel2011}
S.~Har-Peled and B.~Raichel.
\newblock The fr\'{e}chet distance revisited and extended.
\newblock In {\em Proc. 27th Annual ACM symposium on Computational geometry},
  SoCG '11, pages 448--457, New York, NY, USA, 2011. ACM.

\bibitem{lawson1980lectures}
H.~Lawson.
\newblock {\em Lectures on minimal submanifolds}, volume~1 of {\em Mathematics
  lecture series}.
\newblock Publish or Perish, 1980.

\bibitem{MY05}
A.~Maheshwari and J.~Yi.
\newblock On computing {Fr\'{e}chet} distance of two paths on a convex
  polyhedron.
\newblock In {\em European Workshop on Computational Geometry}, pages 41--44,
  2005.

\bibitem{MS92}
R.~B. McMaster and K.~S. Shea.
\newblock {\em Generalization in Digital Cartography}.
\newblock Association of American Cartographers, Washington DC, 1992.

\bibitem{MY91}
K.~Mehlhorn and C.-K. Yap.
\newblock Constructive {Whitney-Graustein Theorem}: {Or} how to untangle closed
  planar curves.
\newblock {\em SIAM J. Comput.}, 20(4):603--621, 1991.

\bibitem{rado}
T.~Rado.
\newblock On plateau's problem.
\newblock {\em Annals of Mathematics}, 31:457--469, 1930.

\bibitem{Rot88}
J.~J. Rotman.
\newblock {\em An {Introduction} to Algebraic Topology}.
\newblock Graduate Texts in Mathematics; 119. Springer-Verlag New York Inc.,
  1988.

\bibitem{Sch92}
H.~Schipper.
\newblock Determining contractibility of curves.
\newblock In {\em Proc. Sympos. on Computational Geometry}, pages 358--367,
  1992.

\bibitem{VY90}
G.~Vegter and C.~K. Yap.
\newblock Computational complexity of combinatorial surfaces.
\newblock In {\em Proc. Sympos. on Computational Geometry}, pages 102--111,
  1990.

\bibitem{Vel01}
R.~C. Veltkamp.
\newblock Shape matching: similarity measures and algorithms.
\newblock In {\em Proc. Shape Modeling International}, pages 188--199, 2001.

\bibitem{Whi37}
H.~Whitney.
\newblock On regular closed curves in the plane.
\newblock {\em Compositio Mathematica}, 4:276--284, 1937.

\end{thebibliography}

%%%%%%%%%%%%%%%%%%%%%%%%%%%%%%%%%%%%%%%%%%%%%%%%%%%%%

\appendix

\section{Proof for Observation \ref{obs:breakptorder}}
\label{appendix:breakptorder}

\parpic[r]{\includegraphics[width=3cm]{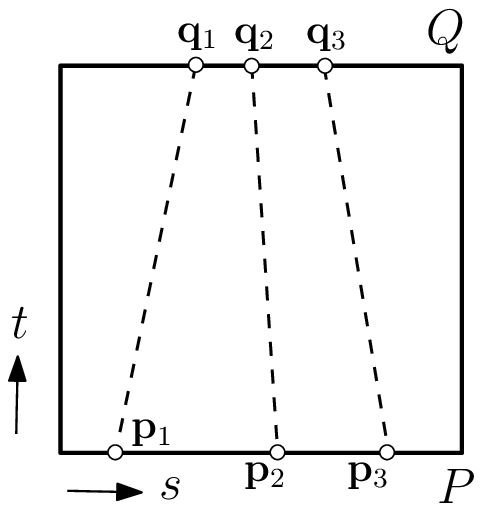}}
Note that $\optmorph$ is a map from $\mysquare \rightarrow M$, where $\mysquare = [0,1] \times [0,1]$ is the unit square and a point $(s,t) \in \mysquare$ will be mapped to $\optmorph_t(s)$. See the right figure for an illustration. (Since $P$ and $Q$ share starting and ending endpoint, the left and right sides of $\mysquare$ should be contracted to a point. We use the square view for simpler illustration.) The top and bottom boundary edges of this square are mapped to $Q$ and $P$, respectively. 
Given an \breakpt{} $\bpt_i$, let $\pin_i$ and $\qin_i$ be the parameters of $\bpt_i$ in $\optmorph_0$ and $\optmorph_1$, respectively; that is, $\optmorph_0(\pin_i) = \optmorph_1(\qin_i) = \bpt_i$. By definition of \breakpt{}s, the pre-image of $\bpt_i$ under the map $\optmorph$ necessarily includes a curve in $\mysquare$ connecting $\pin_i$ on the bottom edge to $\qin_i$ on the top boundary edge of $\mysquare$. 
Since $\bpt_i \neq \bpt_j$, the pre-images of $\bpt_i$ cannot intersect with that of $\bpt_j$. Hence no two such curves can intersect each other, which means that $\pin_i$s must be ordered in the same way as $\qin_i$s. 

\section{Proof for Lemma \ref{lem:sensepreserving}}
\label{appendix:sensepreserving}

\begin{figure*}[tbhp]
\begin{center}
\begin{tabular}{ccccccccc}
\includegraphics[height=2cm]{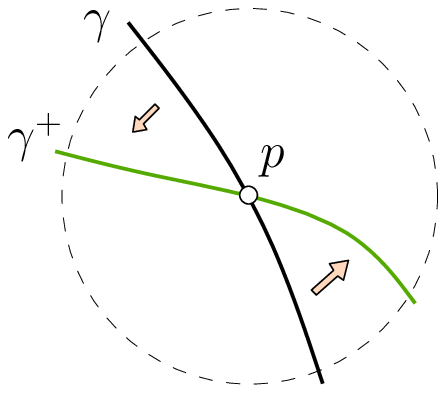} & &
\includegraphics[height=2cm]{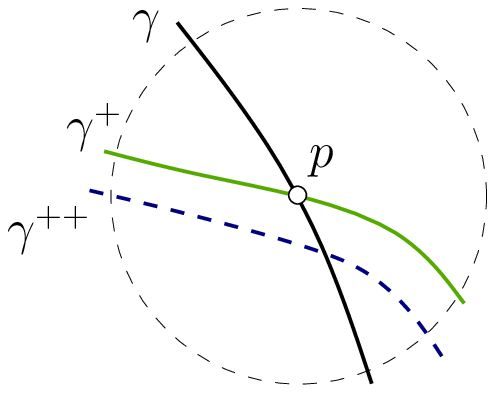} & &
\includegraphics[height=2cm]{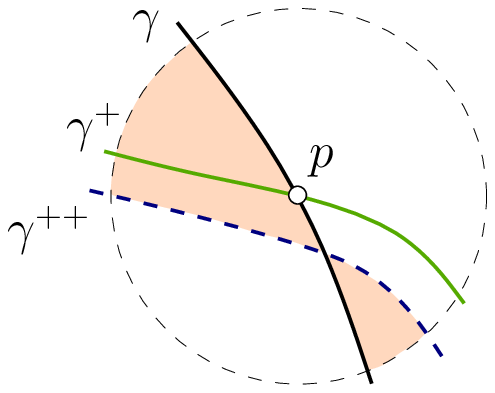} & & 
\includegraphics[height=2cm]{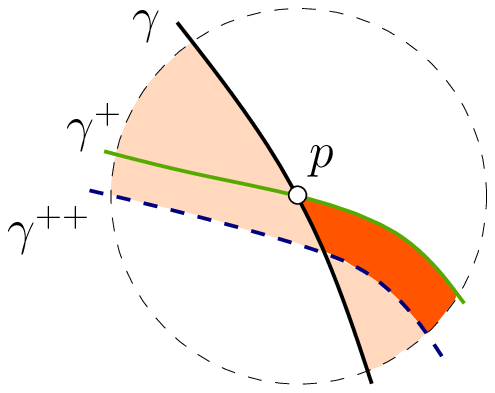} & &
\includegraphics[height=2cm]{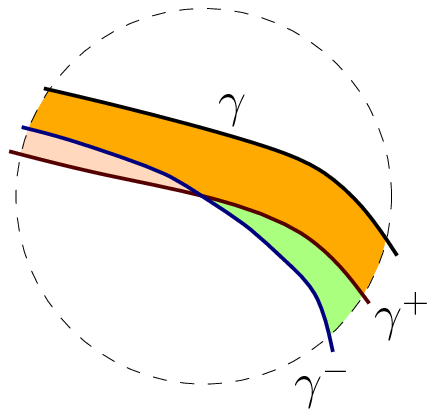}\\
(a) & & (b) & & (c) & & (d) & & (e)
\end{tabular}
\end{center}
\vspace*{-0.15in}
\caption{{\small (a) and (b) $p$ is fixed from $\gamma$ to $\gamma^+$, but not so in $\gamma^{++}$. Sweeping $\gamma$ to $\gamma^{++}$ directly through the shaded region in (c) has a smaller area than first to $\gamma^+$ then to $\gamma^{++}$ (see shaded region in (d)). The darker shaded region in (d) is swept twice. (e) If the deformation changes orientation at $\gamma$, then there is a local fold in the regions swept. }}
\label{fig:deform}
\vspace*{-0.1in}
\end{figure*}

Consider a time $t$ in the homotopy, and let $\gamma = H_t$.  We first show that $H$ deforms $\gamma$ consistently, so that every point on $\gamma$ is either fixed or deforms to the same side of $\gamma$.  

First note that if some portion of $\gamma$ is left sense-preserving at time $t$ and then reverses its direction and becomes right sense preserving at time $t^+$ a small amount later, some portion of the domain has been swept twice.  Hence this homotopy cannot have minimal area, since we can create a smaller one which stops at time $t$ and moves directly to some intermediate curve at a time greater than $t^+$ without sweeping any portion twice. See Figure \ref{fig:deform} (e). 

Now suppose that some portion of $\gamma$ is deforming to one direction and another is morphing in the opposite direction.  Since $H$ is a homotopy and is therefore continuous, this means that there is at least one interval of fixed points between these two regions (which may possibly consist of a single point).  Let $p$ be an extremal point on this interval; see~Figure \ref{fig:deform} (a) for a picture when $p$ is the only fixed point.  In addition, since $p$ is a fixed point but not an anchor point, where know there is some $t^+ = t + dt$ where $p$ is still on $h_{t^+} = \gamma^+$ and another $t^{++} = t^+ + dt$ where $p$ is not on $h_{t^{++}} = \gamma^{++}$.

Now we have several cases to consider.  First, consider if $H$ has directly reversed the direction of either portion of the curve (before $p$ or after $p$), we are in a similar situation to the one previously discussed, since the curve goes from locally forward to locally backward (or vice versa).  In this case, we again know that some area of the domain has been swept twice, which means $\gamma^{++}$ has been swept over once and then was returned to, so we can reduce the area swept by $H$ by reparameterizing $H$ to move directly to $\gamma^{++}$ without passing it and then reversing.  (See~Figure \ref{fig:deform} (b),(c), and (d) for an illustration.)

Now if neither portion directly reverses, then $\gamma^{++}$ must also be deforming to two different directions.  Also, we know that $\gamma^{++}$ must intersect $\gamma$ at some point $q \ne p$, since we are essentially rotating around a central set of fixed points on these curves.  In this case, we can again alter $H$ to attain a smaller area swept by simply sweeping directly from $\gamma$ to $\gamma^{++}$; this will reduce the area since the triangular region in the center bounded by $\gamma, \gamma^+$, and $\gamma^{++}$ will be swept only one time instead of twice.

The claim thus follows, since any homotopy with no anchor points that is not sense preserving cannot be minimal.

%
%\section{Proof for Lemma \ref{lem:negativewn}}
%\label{appendix:negativewn}
%
%Without loss of generality, assume that the map $\amorph$ is right \sensepreserving{}, always deforming an intermediate curve to its right. 
%Consider the time-varying function $F: [0,1] \times \reals^2 \rightarrow \Z$, where $F(t, x) = \wn(x; \amorph_t)$ is the winding number at $x \in \reals^2$ with respect to the curve parameterized by $\amorph_t$. 
%Obviously, $F(0, x) = \wn(x; P \concatenate \text{rev}(Q))$, and $F(1, x) = 0$. 
%During the deformation, $F(t, x)$ changes by either $1$ or $-1$ whenever the intermediate curve sweep over it. 
%Since the \morph{} is right \sensepreserving{}, when an intermediate curve sweeps $x$, $x$ always moves from the left side of the intermediate curve to its right side. Hence the winding number $x$ decreases monotonically. Since in the end, the winding number at each point is zero, $\wn(x; P \concatenate \text{rev}(Q)) = F(0,x) \ge 0$. 
%
%If the map $\amorph$ is left \sensepreserving{}, then a symmetric argument shows that $\wn(x; P \concatenate \text{rev}(Q)) \le 0$ for all $x \in \reals^2$. 

\section{Proof for Lemma \ref{lem:positivewn}}
\label{appendix:positivewn}

\begin{figure*}[tbp]
\begin{center}
\begin{tabular}{ccccccccc}
\includegraphics[height=2.5cm]{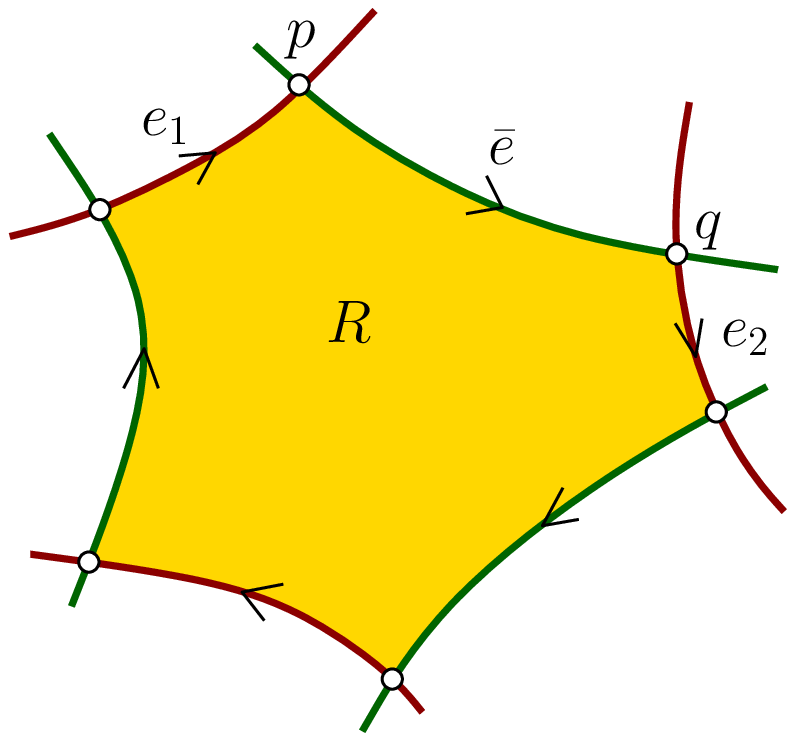} & & 
\includegraphics[height=2.5cm]{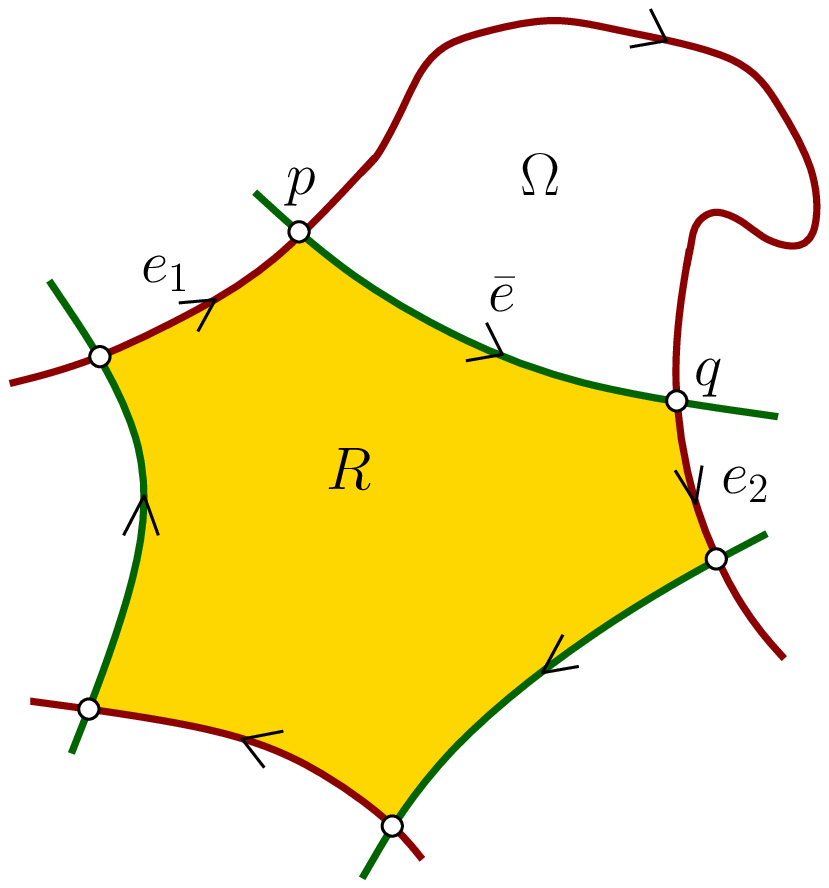} & &
\includegraphics[height=2.5cm]{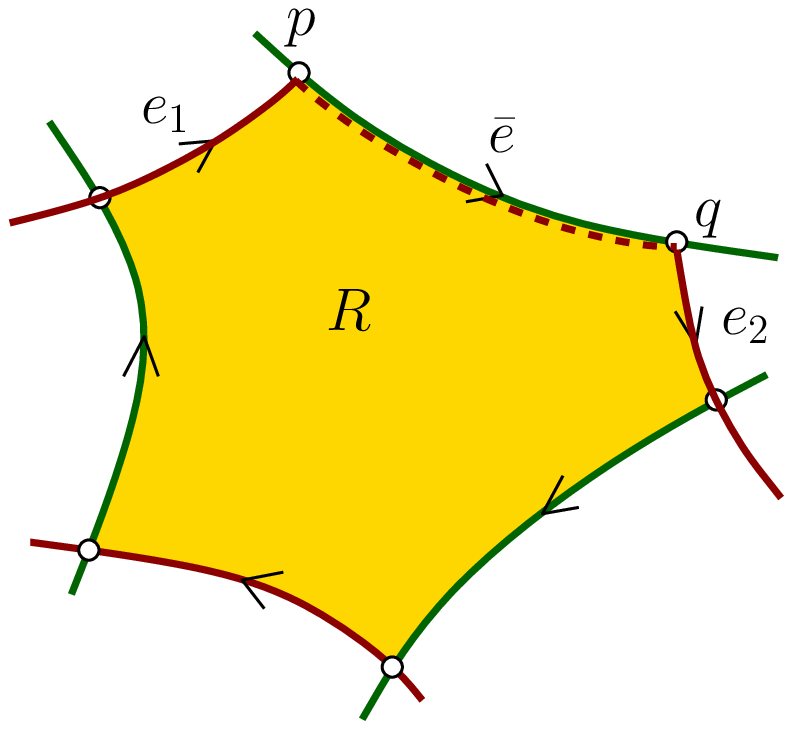} & &
\includegraphics[height=2.5cm]{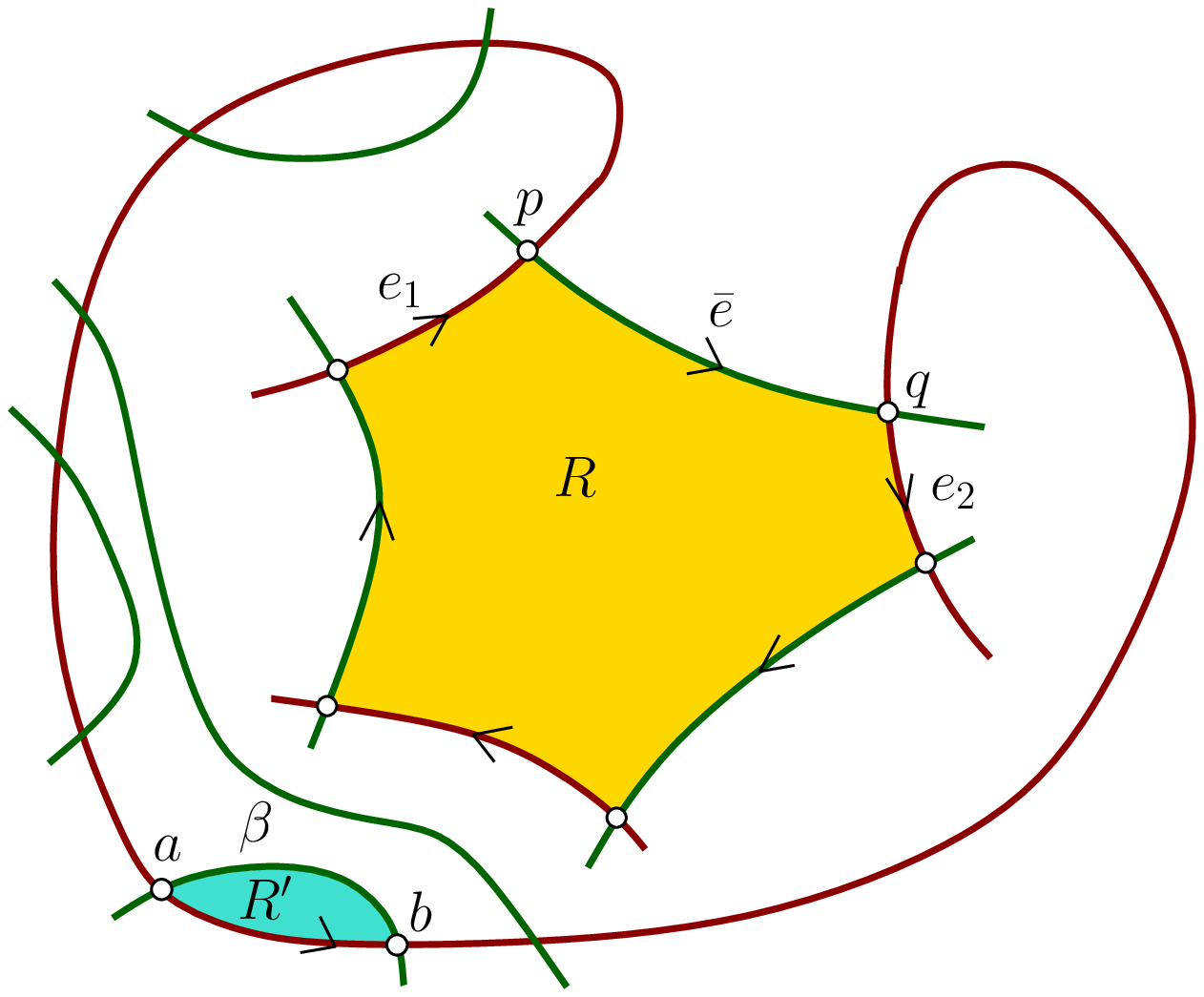} & &
\includegraphics[height=2.5cm]{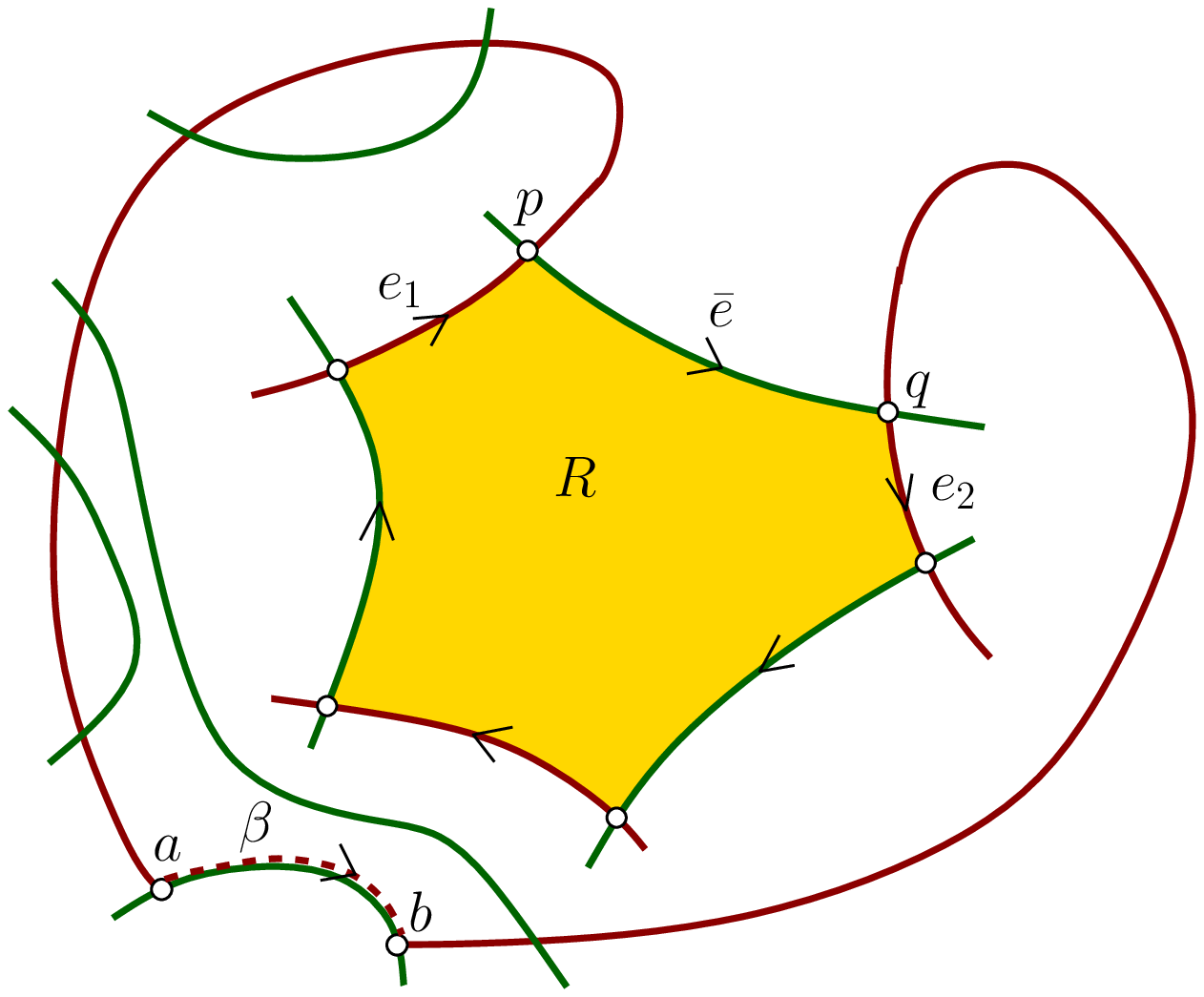} \\
(a) & & (b) & & (c) & & (d) & & (e)
\end{tabular}
%\begin{tabular}{ccccc}
%\includegraphics[height=2.5cm]{./figs/maxwncell2} & &
%\includegraphics[height=2.5cm]{./figs/newmaxwncell3} & &
%\includegraphics[height=2.5cm]{./figs/maxwncell5} \\
%(a) & & (b) & & (c) 
%\end{tabular}
\end{center}
\vspace*{-0.15in}
\caption{{\small (a) The cell $\acell$ with highest positive winding number. It boundary consists of alternating $P$-arcs (red) and $Q$-arcs (green). The two cases of relations between $P[p,q]$ and $\acell$ are shown in (b) and (d), respectively. For case (b), we can deform $P$ to sweep through $\Omega$ as shown in (c), and reduce the number of intersections by $2$. Similarly, for case (d), we can identify any bigon $R'$ and deform $P$ to reduce the number of intersections by $2$ as well. 
}}
\label{appendix:fig:positivewn}
\vspace*{-0.1in}
\end{figure*}
We prove the claim by induction on the number of intersections between $P$ and $Q$. 
The base case is when there is no intersection between $P$ and $Q$. 
In this case, $\acurve$ is a Jordan curve which decomposes $\reals^2$ into two regions, one inside $\acurve$ and one unbound. By orienting $\acurve$ appropriately, every point in the bounded cell has winding number $1$ and the claim follows.  

Now assume that the claim holds for cases with at most $k-1$ intersections. We now prove it for the case with $k$ intersections. Let an \emph{$X$-arc} denote a subcurve of curve $X$. 
Consider the arrangement $\arr(\acurve)$ formed by $\acurve = P \concatenate \text{rev}(Q)$. Since $P$ and $Q$ are simple, every cell in this arrangement has boundary edges alternating between $P$-arcs and $Q$-arcs. Assume without loss of generality that $\acurve$ has all non-negative winding numbers. 
Consider a cell $\acell \in \arr(\acurve)$ with largest (and thus positive) winding number. Since its winding number is greater than all its neighbors, it is necessary that all boundary arcs are oriented consistently as shown in Figure \ref{appendix:fig:positivewn} (a), where the cell $\acell$ (shaded region) lies to the right of its boundary arcs. 

\parpic[r]{\includegraphics[width=3cm]{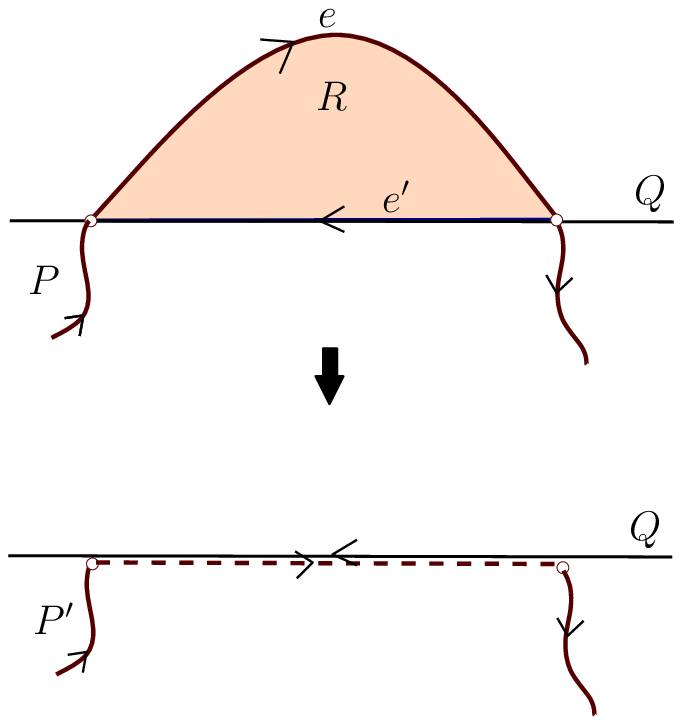}}
If $\acell$ has only two boundary arcs, $e$ from $P$ and $e'$ from $Q$, respectively, then we can morph $P$ to another simple curve $P'$ by deforming $e$ through $\acell$ to $-e'$ (where `$-$' means reversing the orientation). See the right figure for an illustration. 
The area swept by this deformation is exactly the area of cell $\acell$. 
Furthermore, after the deformation, every point $x \in \acell$ decreases their winding number by $1$, and no other point changes its winding number. 
Since points in this cell initially has strictly positive winding number, the resulting curve $\acurve' = P' \concatenate Q$ still has all non-negative winding number. The number of intersections between $P'$ and $Q$ is $k - 2$. By induction hypothesis, $\simC(P', Q) = \totalwn(\acurve')$. Since $\totalwn(\acurve) - \totalwn(\acurve') = \area(\acell)$, we have that $\totalwn(\acurve) = \simC(P', Q) + \area(\acell)$. It then follows from Observation \ref{obs:lowerbound} and the fact $\simC(P, Q) \le \simC(P', Q) + \area(\acell)$ that $\simC(P, Q) = \totalwn(\acurve)$. 

Otherwise, the cell $\acell$ has more than one $P$-arc. Take the $P$-arc $e_1$ with the smallest index along $P$, and let $p$ be the ending endpoint of it. 
Let $e_2$ be the next $P$-arc along the boundary of $\acell$, and $q$ its starting endpoint, and $Q[p,q]$ the $Q$-arc between $e_1$ and $e_2$, denoted by $\bar{e}$ in Figure \ref{appendix:fig:positivewn}. 
Obviously, the subcurve $P[p,q]$ cannot intersect $\acell$, and $P[p,q]$ and $-Q[p,q]$ bound a simple polygon, which we denote by $\Omega$. 
Either $\Omega$ is on the opposite side of the $Q$-arc $\bar{e}$ from the interior of $\acell$ (Figure \ref{appendix:fig:positivewn} (b)), 
or they are on the same side (Figure \ref{appendix:fig:positivewn} (d)).

\vspace*{0.06in}\emph{Case (1): $\acell$ and $\Omega$ are on the opposite side of $\bar{e}$.~} 
In this case, the region $\Omega$ is to the right of the oriented arc $P[p,q]$. 
Note that $P$ does not intersect the interior of $\Omega$; as otherwise, $P$ will either intersect itself or intersect $\bar{e}$, neither of which is possible. Hence only $Q$ can intersect $\Omega$. Since $Q$ is also a simple curve, there is no vertices of $\arr(\acurve)$ contained in the interior of $\Omega$. As a result, every cell of $\arr(\acurve)$ contained in $\Omega$ must have at least one boundary edge coming from $P[p,q]$. 
This implies that each cell contained in $\Omega$ has strictly positive winding number; that is, $\wn(x; \acurve) > 0$ for any $x \in \Omega$. This is because if a cell $\xi \subseteq \Omega$ has winding number $0$, then its neighbor across its boundary on the other side of $P[p,q]$ will have winding number $-1$, as $\Omega$ is to the right of $P[p,q]$. 
This violates the condition that $\acurve$ has all non-negative winding numbers and thus cannot happen. 

We now deform $P$ to $P'$ by sweeping $P[p,q]$ through $\Omega$ to $Q[p,q]$. See Figure \ref{appendix:fig:positivewn} (c). The cost of this sweeping is $\area(\Omega)$ and $\totalwn(\acurve) - \totalwn(P' \concatenate Q) = \area(\Omega)$. $P'$ is still simple, and the number of intersection points between $P'$ and $Q$ is now $k-2$. 
Since $\wn(x; \acurve) > 0$ for any $x \in \Omega$, we have $\wn(x; P' \concatenate Q) \ge 0$ for $x\in \Omega$. No other point will change their winding number after this deformation. Thus the curve $P' \concatenate Q$ has all non-negative winding numbers as well. Hence by induction hypothesis, we have that $\simC(P', Q) = \totalwn(P' \concatenate Q)$. 
Since $\simC(P, Q) - \simC(P', Q) \le \area(\Omega)$ and $\totalwn(\acurve) - \totalwn(P' \concatenate Q) = \area(\Omega)$, it then follows from Observation \ref{obs:lowerbound} that $\simC(P, Q) = \totalwn(\acurve)$. 

\vspace*{0.06in}\emph{Case (2): $\acell$ and $\Omega$ are both from the same side of $\bar{e}$.~} We now consider the remaining case as shown in Figure \ref{appendix:fig:positivewn} (d). 
Take the unbounded region $\overline{\Omega} := \reals^2 \setminus \Omega$ which is the complement of $\Omega$. This unbounded region lies to the right of the oriented curve $P[p,q]$. 
Since both $P$ and $Q$ are simple, only $Q$ can intersect $\Omega$ and there is no vertices of $\arr(\acurve)$ contained in the interior of $\overline{\Omega}$. 
First, observe that it is not possible that $\overline{\Omega} \cap Q = \emptyset$. This is because otherwise, $\overline{\Omega}$ is the unbounded face of $\arr(\acurve)$ and thus the winding number for all points in $\overline{\Omega}$ is $0$. This however is not possible as this will imply that any point $y$ to the immediate left of $P[p,q]$ has winding number $-1$, violating our assumption that all cells in $\arr(\acurve)$ have consistent (non-negative) winding numbers. 

Hence $\overline{\Omega} \cap Q \neq \emptyset$, and there are set of arcs from $Q$ intersecting $P[p,q]$.  Then there must exist a bigon cell $R'$ bounded by only two arcs, one $P$-arc $P[a,b] \subseteq P[p,q]$ and a $Q$-arc $\beta$. See Figure \ref{appendix:fig:positivewn} (d). 
Similar to the argument from the previous paragraph, we can show that points in $R'$ must have strictly positive winding number. 
Now deform $P$ to $P'$ by sweeping $P[a,b]$ through $R'$ to $\beta$ as shown in Figure \ref{appendix:fig:positivewn} (e). $P'$ is still simple, the number of intersection points between $P'$ and $Q$ is now $k-2$. Only points in $R'$ reduce their winding number by $1$, and the resulting arrangement still has consistent winding numbers. As such,  by induction hypothesis, we have that $\simC(P', Q) = \totalwn(P' \concatenate Q)$. 
Since $\simC(P, Q) - \simC(P', Q) \le \area(R')$ and $\totalwn(\acurve) - \totalwn(P'\concatenate Q) = \area(R')$, it then follows from Observation \ref{obs:lowerbound} that $\simC(P, Q) = \totalwn(\acurve)$.

% By yusu June 28 2009.

\section{Cycles in the Plane}
%%%%%%%%%%%%%%%%%%%%%%%%
%\subsection{Cycles in the plane}
\label{appendix:sec:planarcycles}

We now consider the case where we have two simple cycles $P$ and $Q$ in the plane.  We have the following characterization: 

\begin{lemma}
\label{lem:anchorcycles}
If the two simple cycles $P$ and $Q$ intersect, then there is an anchor point in the optimal homotopy between them. 
\end{lemma}

\begin{proof}
Suppose that $P$ and $Q$ intersect but there is no anchor point in the optimal homotopy $H^*$.  By Lemma~\ref{lem:sensepreserving}, we know that $H^*$ must be sense preserving.  However, this means that $H^*$ continually moves one curve to the other in one local direction, which means that one curve must be entirely contained within the other, contradicting the assumption that they intersect. 
\end{proof}

At this point, the algorithm for cycles which intersect each other reduces to the one for curves: 
if we know which intersection point between $P$ and $Q$ is the \breakpt{}, 
we can simply ``break" the cycles into two curves at this point; 
this will become the start and end point for each of the two curves.  
Hence our algorithm for cycles will take a multiplicative factor of $O(I)$ extra time than the algorithm for curves, since we need to try each possible intersection point as the required anchor point. 

It then follows  from Theorem \ref{thm:plane} that: 

\begin{corollary}
Given two polygonal cycles $P$ and $Q$ in the plane of $n$ total complexity and with $I>0$ intersection points, we can compute the optimal \morph{} and its area in $O(I(I^2 \log I + n \log n))$ time.  
\end{corollary}

The remaining case is that when the two polygonal cycles $P$ and $Q$ are disjoint. 
If one of the cycle contains the other, then the area of the optimal homotopy is simply the area sandwiched between these two simple cycles. This can be computed in $O(n\log n)$ time easily. 

However, if the cycles are disjoint but neither contains the other, then in a sense the ``optimal" free homotopy between them will simply be the area bounded by each curve, since the homotopy can collapse each curve separately to a point and then deform the points to each other. Indeed, that the sum of area bounded by the two Jordan cycles $P$ and $Q$ is the minimum possible homotopy area follows from a similar argument as the proof of Observation \ref{obs:lowerbound}. 
However, in this case, the free homotopy described above is not regular since it collapses a curve to a single point. 
Nevertheless, one can argue that there exists a sequence of regular homotopies whose areas converge to this sum. 
In other words,  the optimal area homotopy is still well-defined (as the \emph{infinum} of the area of regular free homotopies between $P$ and $Q$), although there does not exist a regular homotopy to achieve this optimal area. (This is analogous to similar issues that arise in the general case, and is the reason for introducing more restricted integrals in the mathematical literature~\cite{lawson1980lectures}.)

\section{Missing Details for the Sphere Case}
\label{appendix:sphere}

\begin{obs}
Given a closed curve $\Gamma$ and any two points $\z, \w \in \sphere$, we have that: $\wn(x; \w) = \wn(x; \z) + \wn(\z; \w)$ (all winding numbers are w.r.t the curve $\Gamma$). 
In particular, for any two points $\z_1, \z_2$ from the same cell of $\arr(\Gamma)$, we have that $\wn(x; \z_1) = \wn(x; \z_2)$ for all $x \neq \z_1, \z_2$. 
\label{obs:samecellwn}  
\end{obs}
\myproofbegin
Let $\gamma(x, y)$ be a path connecting point $x$ to $y$. Note that the concatenation between $\gamma(x, \z)$ and $\gamma(\z, \w)$ is a path from $x$ to $\w$. Since $\wn(x; \w)$ is simply the summed signed crossing number of any path from $x$ to $\w$ with respect to $\Gamma$, the claim follows immediately. 
\myproofend

\subsection{Proof for Observations \ref{obs:onecell}}
\label{appendix:obs:onecell}

\paragraph{Proof of Observation \ref{obs:onecell}}
Suppose $x$ and $y$ are two points from the interior of $\acell$ such that $\optmorph$ sweeps through $x$, but not $y$. Connect $x$ with $y$ by any path $\gamma$ in the interior of $\acell$. This path has to intersect the boundary of the region swept by $\optmorph$, and let $z$ be one such intersection point on $\gamma$. Obviously, there is a local fold in the optimal \morph{} as it sweeps through $z$; namely, some intermediate curve will touch $z$ and immediately trace back. Thus the input \morph{} $\optmorph$ cannot be \sensepreserving{}. Contradiction. Hence $\optmorph$ sweeps $y$ as well. 
%\myproofend

%\paragraph{Proof of Observation \ref{obs:samecellwn}}
%Let $\gamma(x, y)$ be a path connecting point $x$ to $y$. Note that the concatenation between $\gamma(x, \z)$ and $\gamma(\z, \w)$ is a path from $x$ to $\w$. Since $\wn(x; \w)$ is simply the summed signed crossing number of any path from $x$ to $\w$ with respect to $\Gamma$, the claim follows immediately. 

\subsection{Proof for Lemma \ref{lem:infpointexists}}
\label{appendix:lem:infpointexists}

We prove the lemma by induction on the number of intersections between $P$ and $Q$. 
When there is no intersection between $P$ and $Q$ (other than the common endpoints), the Jordan curve $P \concatenate \text{rev}(Q)$ divides the sphere into two connected components, and the optimal \morph{} is the smaller area of the two. The lemma holds for this base case. 

Now assume that the lemma holds for $P$ and $Q$ with at most $k$ intersection points. We wish to show the result for the case where $P$ and $Q$ have $k+1$ intersection points. Since $\optmorph$ has no \breakpt{}s, this optimal \morph{} is \sensepreserving{} by Lemma \ref{lem:sensepreserving}. Assign an orientation to the closed curve $C= P\concatenate Q$ so that locally, every point on the curve $P$ will continuously deform to its right during the optimal \morph{}. 
Now pick an \emph{arbitrary} point $\z$ not on $P$ and $Q$, and compute the winding number for each cell of $\arr(P + Q)$ w.r.t. $\z$. 
Take the cell $\acell$ with the largest winding number. 
We assume that $\z \notin \acell$. 
% if $\z \in \acell$, it turns out that we can the re-identify a choice of $\z$ such that this does not happen. 
Suppose this is not the case and that $\z \in R$. Then we show that we can change the choice of $\z$ to make this hold. 

Specifically, if $\z \in R$, then the cell $R$ must have winding number $0$. Now take the cell $R'$ of $\arr(P + Q)$ with the smallest winding  number, and let $\w$ be a point from $R'$. Obviously, $\wn(\w; \z) \le \wn (x; \z) \le 0$ for any $x \in \sphere$. 
Now we consider the winding numbers w.r.t. to $\w$ instead of $\z$. 
By Observation \ref{obs:samecellwn} we have that $\wn(x; \w) = \wn(x; \z) + \wn(\z; \w)$. On the other hand, we have that $\wn(\z; \w) = - \wn(\w; \z)$.  Hence $\wn(\z; \w) \ge \wn(x; \z) \ge 0$ for any $x \in \sphere$. In other words, for this new choice of point $\w$, we have that $R$ still has the largest winding number and in this case, $\w \notin R$. 

%\parpic[r]{\includegraphics[height=4cm]{./figs/optwndsphere1}}
%\begin{wrapfigure}{r}{3.5in}
\begin{figure}[htbp]
\begin{center}
\begin{tabular}{ccc}
\includegraphics[height=3cm]{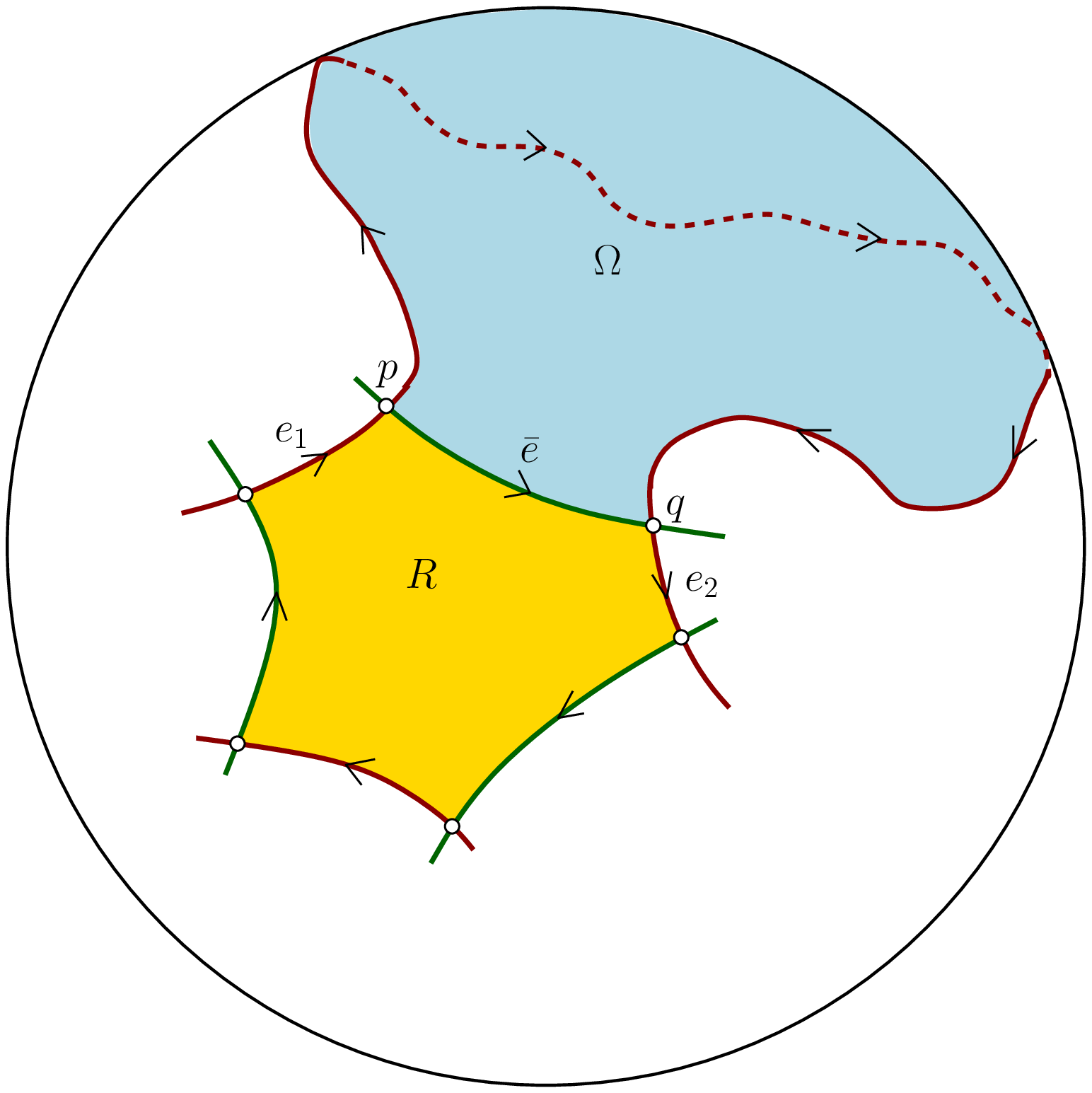} & \hspace*{0.05in} &
\includegraphics[height=3cm]{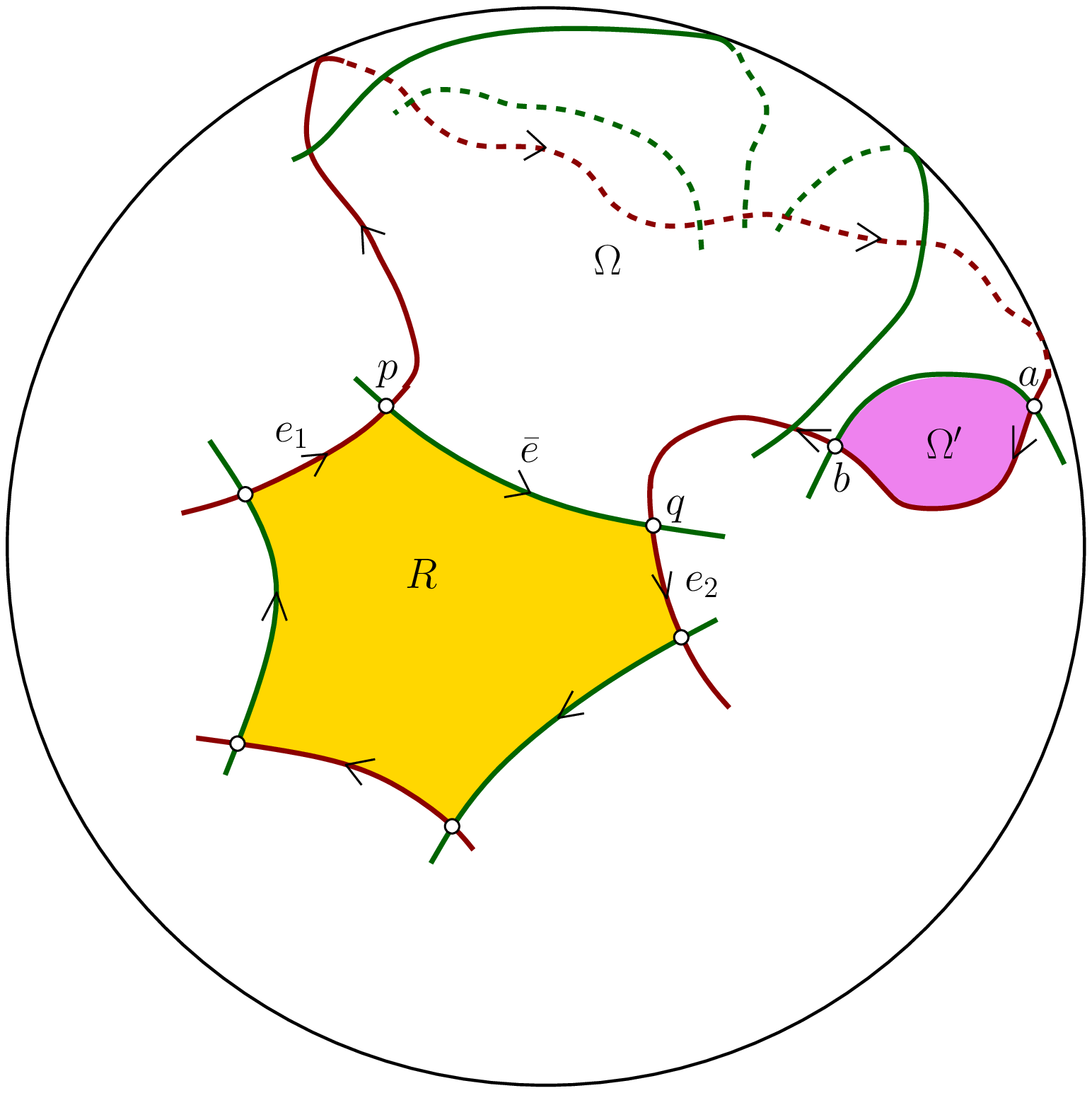} \\
(a) & & (b)
\end{tabular}
\end{center}
\caption{(a) Take $\Omega$ as the region to the left of $Q[p,q]$. (b) There always exists a bigon $\Omega'$ (possibly $\Omega$ is no $Q$-arc intersects $P[p,q]$. 
}
\label{fig:bigon}
\end{figure}
%\end{wrapfigure}
Similar to the proof of Lemma \ref{lem:positivewn}, the boundary of this cell consists of alternating arcs from $P$ and from $Q$, and they necessarily have the orientation as shown in Figure \ref{fig:positivewn} (a) (otherwise, one of the neighboring cell if $\acell$ will have a larger winding number). 
Choose the $P$-arc $e_1$ that appears earliest along $P$, with $p$ being its ending endpoint. Let $P[q]$ be the next intersection between $P$ and $\acell$. We have that $P[p,q]$ and $\bar{e} = Q[p,q]$ do not intersect each other. The Jordan curve $P[p,q]\concatenate (\text{rev}(Q[p,q]))$ bounds two regions on the sphere (instead of a bounded one and an unbounded one in the planar case as shown in Figure \ref{fig:positivewn} (b) and (d)). We consider the region that lies to the right of $P[p,q]$ (thus left of $\bar{e}$), and denote it by $\Omega$. See Figure \ref{fig:bigon} (a). Since $\acell$ is to the right of $\bar{e}$, $\Omega \cap \acell = \emptyset$. 
%
%\parpic[r]{\includegraphics[height=4cm]{./figs/optwndsphere2}}
As $P$ is simple, only $Q$ can potentially intersect the cell $\Omega$. 
Hence there always exists a bigon $\Omega'$ in $\Omega$ which is a cell in $\arr(P+Q)$. 
See Figure \ref{fig:bigon} (b). 
Note that it is possible that $\Omega' = \Omega$. 
Let $P[a,b]$ denote the $P$-arc that bounds the bigon $\Omega'$. 
Let $P'$ be a new curve obtained by replacing $P[a,b]$ with (slightly above) $Q[a,b]$. 
Since $P$ deforms always to the right in the optimal homotopy, and in the end, $P[a,b]$ needs to deform to some portion of $Q$ (which is not necessarily $Q[a,b]$ though), one can show that there is an optimal \morph{} between $P$ and $Q$ that consists of first sweeping $P[a,b]$ to $Q[a,b]$ through $\Omega'$, and then optimally morph $P'$ to $Q$.  On the other hand, by the induction hypothesis, there is an optimal \morph{} $\amorph'$ from $P'$ to $Q$ that does not sweep some point, say $\z_1$ in $\sphere$. 
There are now two cases: 
\begin{itemize}\denselist
\item[(i)] If $\z_1 \in \sphere - \Omega'$, then there is an optimal \morph{} from $P$ to $Q$ that does not sweep $\z_1$ as well. The induction step then holds and the claim follows.  
\item[(ii)] Otherwise, $\z_1 \in \Omega'$. Consider the cell $R' \in \arr(P' + Q)$ that contains $\z_1$. Note that $R' \cap (\sphere - \Omega') \neq \emptyset$, as there is no vertices of $\arr(P' + Q)$ contained neither on nor inside $\Omega'$. Hence $R'$ must also contain some point, say $\z_2$, that is outside of $\Omega'$. It then follows from Observation \ref{obs:onecell} that $\z_2$ is not swept either. This leads us back to case (i), and the induction step again holds. 
\end{itemize}
The claim then follows by induction.

%\section{Missing Details from Proof of Lemma \ref{lem:infpointexists}}
%\label{appendix:lem:infpointexists}
%
%Recall that we start by picking an arbitrary point $\z$ not on $P$ nor $Q$, and computing the winding number for all cells in $\arr(P + Q)$, and $R$ is the cell with the largest winding number. 
%We assume that $\z \notin R$. Now suppose this is not the case and that $\z \in R$. Then we show that we can change the choice of $\z$ to make this hold. 
%Specifically, if $\z \in R$, then the cell $R$ has winding number $0$. Now take the cell $R'$ of $\arr(P + Q)$ with the smallest winding  number, and let $\w$ be a point from $R'$. Obviously, $\wn(\w; \z) \le \wn (x; \z) \le 0$ for any $x \in \sphere$. 
%Now we consider the winding numbers w.r.t. to $\w$ instead of $\z$. 
%By Observation \ref{obs:samecellwn} we have that $\wn(x; \w) = \wn(x; \z) + \wn(\z; \w)$. On the other hand, we have that $\wn(\z; \w) = - \wn(\w; \z)$.  Hence $\wn(\z; \w) \ge \wn(x; \z) \ge 0$ for any $x \in \sphere$. In other words, for this new choice of point $\w$, we have that $R$ still has the largest winding number and in this case, $\w \notin R$. 

\subsection{Details of Algorithm for Sphere Case}
\label{appendix:algsphere}

\paragraph{Dynamic programming framework.}
Similar to the planar case, let $\Ipt_0, \ldots, \Ipt_I$ denote the intersection points between $P$ and $Q$, ordered by their indices along $P$, with $\Ipt_0$ and $\Ipt_I$ being the beginning and ending points of $P$ and $Q$. 
Let $T(i)$ denote the optimal homotopy area between $P[\mathbf{0}, \Ipt_i]$ and $Q[\mathbf{0}, \Ipt_i]$, and $C[i,j]$ the closed curve formed by $P[\Ipt_i, \Ipt_j] \concatenate Q[\Ipt_j, \Ipt_i]$. 
However, now we say that a pair of indices $(i,j)$ is \emph{valid} as long as $\Ipt_i$ and $\Ipt_j$ have the same order along $P$ and along $Q$. This is different from the definition of valid pairs of indices as in the planar case. 

Specifically, for any closed curve $\gamma$, it turns out that $\gamma$ can always have consistent winding number for \emph{some} choices of the point of infinity $\z$: that is, there always exists $\z \in \sphere$ such that $\wn(x, \z; \gamma)$ is consistent for all $x \in \sphere_{\z}$. We call such choices of $\z$ \emph{consistent representatives w.r.t. $\gamma$}. 
Let $\totalwn^*(\gamma) := \min_{\z} |\totalwn(\gamma; \sphere_\z)|$ where $\z$ ranges over all possible choices of consistent representatives w.r.t. $\gamma$. 
Then, Lemma \ref{lem:positivewn} and \ref{lem:infpointexists} imply that if there is an optimal homotopy between $P[\Ipt_i, \Ipt_j]$ and $Q[\Ipt_i, \Ipt_j]$ with no \breakpt{}s, then the optimal homotopy area is $\totalwn^*(\gamma)$. 
However, different from the planar case, $\totalwn^*(C[i,j])$ is defined for all valid pairs of $i$, $j$s, and it may not in general be the optimal homotopy area for $P[\Ipt_i, \Ipt_j]$ and $Q[\Ipt_i, \Ipt_j]$. 
We now have the following recurrence: 
\begin{eqnarray*} 
T(i) = 
\begin{cases} 
  0,  & \mbox{if } i == 0 \\
  \min_{j < i \mbox{~and $(j, i)$ is valid}} ~\{ ~\totalwn^*(C[j, i]) + T(j) ~\}, & \mbox{ otherwise} 
\end{cases}
\end{eqnarray*}
As before, the final goal is to compute $T(I) = \simC(P,Q)$. 

\paragraph{Computing $\totalwn^*$s.}
Here we describe how to compute $\totalwn^*(C[i,j])$ efficiently. 
%$\simC(i,j) := \simC(P[\Ipt_i, \Ipt_j], Q[\Ipt_i, \Ipt_j])$ efficiently. 
Specifically, we show how to compute all $\totalwn^*(C[r, j])$s for all $j > r$ in $O(I)$ time, for any fixed $r$, after $O(n\log n + I\log I + \tricomplexity) = O(n\log n + \tricomplexity)$ preprocessing. 

\parpic[r]{\includegraphics[width=4cm]{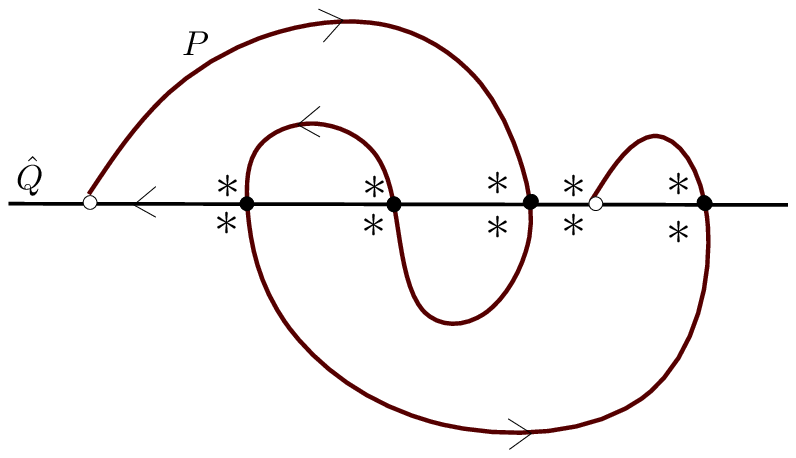}}
First, let us choose the representatives by taking two points around each intersection points $\Ipt_i$ between $P$ and $Q$ as shown in the right figure. 
Consider only those representatives to the right of $Q$ (which are those above $Q$ in the right figure), and denote them by $Z = \{\z_1, \ldots, \z_I \}$. $Z$ is sorted by their indices along $Q$. 
Those to the left of it can be handled similarly. 
The first observation is that for any two consecutive \representative{}s, $\wn(x; \z_i) - \wn(x; \z_{i+1})$ is $1$ or $-1$, depending on the orientation of $P$-arc that separating them. 

Now to compute which $\z_i$ will give \allpositive{} winding numbers, we first compute the winding number of each cell in $\arr(P+ Q)$ for $\z_1$. Next, take the cells $R_1$ and $R_2$ with minimum and maximum winding numbers, and assume that $\z_{i_1}$ and $\z_{i_2}$ are their representatives. If there are more than one cells with largest (or smallest) winding numbers, just pick an arbitrary one. By Observation \ref{obs:samecellwn}, the closed curve $C[r,j] := P[\Ipt_r, \Ipt_j]\concatenate Q[\Ipt_r, \Ipt_j]$ has all non-negative winding number w.r.t $\z_{i_1}$, and all non-positive winding number w.r.t. $\z_{i_2}$. Hence we simply compute the total winding number $\totalwn(C[r,j]; \z_{i_1})$ and $\totalwn(C[r,j]; \z_{i_2})$, and return the one with the smaller absolute value as $\totalwn^*(C[r,j])$. We refer to the indices $i_1$ and $i_2$ as the \emph{\Wmin{} and \Wmax{} indices}, respectively, and these two total winding numbers as \emph{\valid{} total winding numbers}. Basically, by Lemma \ref{lem:validwn} below, the smaller of the absolute values of the two \valid{} total winding numbers is simply the best \energy{} to deform $P[\Ipt_r, \Ipt_j]$ to $Q[\Ipt_r, \Ipt_j]$ without using \breakpt{}s.  
This improves the time complexity of computing each $\totalwn^*(C[r,j])$ to $O(I)$ time, instead of the naive $O(In)$ time by computing all $\totalwn(C[r,j]; \z_i)$s, for $i \in [1,I]$, from scratch. 

\begin{lemma}
Given an arbitrary oriented (not necessarily simple) curve $C = P' \concatenate \text{rev}(Q')$ on $\sphere$, let $\arr(C)$ be the arrangement of $C$, and $Z = \{z_1, \ldots, z_k\}$ a set of representative points from each cell in $C$. Pick an arbitrary point, say $z_1$, and compute the winding number of $C$ w.r.t. each $z_i$. Let $i_1$ and $i_2$ be the wn-min and wn-max indices. Then the best cost to deform $P'$ to $Q'$ with no anchor point is $\min \{|\totalwn(C; z_{i_1})|, |\totalwn(C; z_{i_2})| \}$. 
\label{lem:validwn}
\end{lemma}
\myproofbegin
First, call a point $z$ \emph{valid} if $\wn(x; z, C)$ is consistent for all $x \in \sphere$. 
The optimal cost to deform $P'$ to $Q'$ with no anchor point is $\min_{\text{valid~}z \in Z} |\totalwn(C; z)|$. 
For any base point $z$, note that by Observation \ref{obs:samecellwn}, we have $\wn(x; z) = \wn(x; z_1) - \wn(z; z_1)$. 
In order for the winding number to be consistent, we need that either $\wn(x; z) \ge 0$ for any $x\in \sphere$, or $\wn(x; z) \le 0$ for any $x \in \sphere$. 
Assume it is the former case. 
Then $\wn(z; z_1) \le \wn(x; z_1)$ for all $x \in \sphere$, implying that $\wn(z; z_1) = \wn(z_{i_1}; z_1)$. 
Furthermore, note that 
$$\totalwn(C; z) = \int_\sphere \wn(x; z) dx = \int_\sphere [\wn(x; z_1) - \wn(z; z_1)] dx = \int_\sphere [\wn(x; z_1) - \wn(z_{i_1}; z_1)] dx = \totalwn(C; z_{i_1}). $$
If it is the latter case, 
then $\wn(z; z_1) \ge \wn(x; z_1)$ for all $x \in \sphere$, implying that $\wn(z; z_1) = \wn(z_{i_2}; z_1)$.
In this case we have that 
$$\totalwn(C; z) = \int_\sphere \wn(x; z) dx = \int_\sphere [\wn(x; z_1) - \wn(z; z_1)] dx = \int_\sphere [\wn(x; z_1) - \wn(z_{i_2}; z_1)] dx = \totalwn(C; z_{i_2}). $$
The optimal cost $\min_{\text{valid~}z \in Z} |\totalwn(C; z)|$ is thus achieved by 
the smaller one of the absolute value of $\totalwn(C; z_{i_1})$ and $\totalwn(C; z_{i_2})$. 
\myproofend

\vspace*{0.08in}\noindent To further improve the time complexity, we will start with $C[r,r+1]$, and update the winding number information in each cell as well as the \valid total winding numbers, as we traverse $P$ and pass through each intersection point $\Ipt_i$. To this end, we use the same range tree data structure as in Section \ref{subsec:algorithm}. Specifically, we use this data structure to maintain the winding number information w.r.t a fixed based point $\z_1$. The \Wmin{} and \Wmax{} indices can be easily maintained by storing at each internal node the minimum and maximum winding number within its subtree. We can also maintain the total winding number w.r.t. the base point $\z_1$. The time complexity for updates remains the same as before (i.e, $O(\log I)$ time per update). 

The remaining question is to compute the \valid{} total winding numbers as $i$ increases. 
Let $A$ denote the total area of topological sphere $\sphere$. 
First, observe that for a fixec curve $C$, by Observation \ref{obs:samecellwn}, we have $\wn(x; \z, C) = \wn(x; \z_1, C) - \wn(z; \z_1, C)$ with respect the fixed based point $\z_1$. 
Hence we have that: 
\begin{align}
\totalwn(C; z) &= \int_\sphere \wn(x; z, C) dx = \int_\sphere \wn(x; \z_1, C) dx - A \cdot \wn(z; \z_1, C) \nonumber \\
&= \totalwn(C; \z_1) - A \cdot \wn(z; \z_1, C). 
\label{eqn:computetw}
\end{align}
Assume that $i_1$ is the \Wmin{} index and $i_2$ is the \Wmax{} index. 
Hence we can compute $\totalwn(C; \z_{i_1})$ and $\totalwn(C; \z_{i_2})$ in $O(1)$ time using Eqn (\ref{eqn:computetw}), since $\totalwn(C; \z_1)$, $\wn(\z_{i_1}; \z_1, C)$ and $\wn(\z_{i_2}; \z_2, C)$ are all maintained as $C$ changes from $C[r,u]$ to $C[r,u+1]$. 

% enclosed by yusu March 21, 2013. 
%The remaining question is to maintain the \valid total winding numbers. Consider the \valid total winding  number corresponding to \Wmin{} index (that for \Wmax{} index can be maintained similarly). 
%Let $A$ denote the total area of topological sphere $\sphere$. First, observe that for any $C$, 
%$\totalwn(C; \z_{i+1})$ is simply $\totalwn(C; \z_i) + \alpha A$, where $\alpha = 1$ or $-1$, depending whether $P$ and $Q$ forms a positive or negative crossing at $\Ipt_i$ (a positive crossing means that $P$ rotates in a clockwise manner to $Q$ at the intersection point).  
%Let $\wnmin_u$ denote the \Wmin{} index for $C[r,u]$, and $V[u]$ denote $\totalwn(C[r,u]; \wnmin_u)$. 
%Note that as we pass $\Ipt_{u+1}$, the \Wmin{} index $\wnmin_{u+1}$ either stays the same, or change from $\wnmin_u$ by $1$ or $-1$. 
%Let $V'$ be the total winding number w.r.t. $\wnmin_u$ after we pass through intersection $\Ipt_{u+1}$.  
%If $\wnmin_u = \wnmin_{u+1}$, then $V[u+1] = V'$. 
%Otherwise, $V[u+1] = V' + (\wnmin_{u+1} - \wnmin_u) A$. 
%Hence the update of \valid{} total winding numbers takes only $O(1)$ time per intersection point for each $u > r$. 

Putting everything together, with $O(n\log n + \tricomplexity)$ pre-processing time, we can compute all $\simC(r, j)$s for all $j > r$  for any fixed $r$, in $O(I\log I)$ time, and thus computing all $\totalwn^*(C[r,u])$s for all $r \in [1,I]$ and all $u < r$, in $O(I^2 \log I)$ total time. Putting everything together, the dynamic programming problem can be solved in $O(n\log n + I^2\log I + \tricomplexity)$ total time.

\end{document}